\documentclass[11pt]{article}

\usepackage[affil-it]{authblk}
\usepackage{caption}
\usepackage{subcaption}
\usepackage{bbm}
\usepackage{amsmath,amssymb,amsthm}
\usepackage{graphicx}
\usepackage{fullpage}
\usepackage{color}
\usepackage{multirow}
\usepackage[colorlinks=true, citecolor=blue]{hyperref}

\def\Cend{C^{\rm end}}
\def\Cmax{C^{\rm max}}
\def\Pmax{P^{\rm max}}

\newcommand{\BE}{\begin{equation}}
\newcommand{\EE}{\end{equation}}

\newtheorem{lem}{Lemma}

\newtheorem{pred}{Prediction}

\makeatletter
\def\blfootnote{\xdef\@thefnmark{}\@footnotetext}
\makeatother

\newcommand \Z  {{\mathbb Z}}
\newcommand \dist {{\rm dist}}
\newcommand  \R  {{\mathbb R}}
\newcommand \Half {{\mathbb H}}
\newcommand \C {{\mathbb C}}
\newcommand \SLE  {{${\rm SLE}_{8/3}$}}
\newcommand \Cont {{\rm Cont}}
\newcommand \saws {{\mathcal W}}
\newcommand \bridges {{\mathcal B}}
\newcommand \polygons {{\mathcal P}}
\newcommand \strip {{{\mathcal S}}}
\newcommand \E {{\mathbb E}}
\renewcommand \Im {{\rm Im}}
\renewcommand \Re {{\rm Re}}

\begin{document}
\title{Compressed self-avoiding walks, bridges and polygons}
\date{}

\author{Nicholas R Beaton}
\affil{Department of Mathematics and Statistics, University of Saskatchewan, Saskatoon, Canada}
\author{Anthony J Guttmann}
\author{Iwan Jensen}
\affil{School of Mathematics and Statistics, University of Melbourne, Victoria 3010, Australia}
\author{Gregory F Lawler}
\affil{Department of Mathematics, University of Chicago, Chicago, Illinois, USA}

\maketitle

\hspace{1.7cm} {\em Dedicated to R.J. Baxter, for his 75th birthday.}

\begin{abstract}We study various self-avoiding walks (SAWs) which are constrained to lie in the upper half-plane and are subjected to a compressive force. This force is applied to the vertex or vertices of the walk located at the maximum distance above the boundary of the half-space. In the case of bridges, this is the unique end-point. In the case of SAWs or self-avoiding polygons, this corresponds to all vertices of maximal height.  We first use the conjectured
  relation with 
  the Schramm-Loewner evolution to predict the form of the
  partition function including the values of the exponents, 
  and then we use
     series analysis to test these
     predictions.\blfootnote{Email: {\tt n.beaton@usask.ca}, {\tt guttmann@unimelb.edu.au}, {\tt i.jensen@ms.unimelb.edu.au}, {\tt lawler@math.uchicago.edu}}

\end{abstract}
Keywords: self-avoiding walk; Schramm-Loewner evolution;
 bridges; compressive force

\section{Introduction}\label{sec:intro}

Self-avoiding walks (SAWs) were initially introduced~\cite{Flory1949Configuration,Orr1947Statistical} as a model of long linear polymer chains. Since that time they have been explored both as a polymer model in statistical mechanics and as an independent problem of interest to combinatorialists and theoretical computer scientists.

In the simplest case, one considers walks on the edges of a lattice,
starting from a fixed origin and forbidden from visiting a vertex more than once. We will denote by $c_n$ the number of such objects with exactly $n$ steps. 
A simple subadditivity argument~\cite{Hammersley1957Percolation} shows that  
\[  \lim_{n\to\infty} \frac1n \log c_n =\beta  := \inf_{n \to \infty}  \frac1n \log c_n
 .\]

  The quantity $ e^{\beta} $ is known as the \emph{connective constant},
  and depends on the lattice in question.
The   connective constant 
 is unknown for all but one regular lattice, though numerical estimates have been computed for a number of cases. Of interest here is the current best (nonrigorous) estimate 
 for the connective
 constant  of SAWs on the square lattice \cite{Clisby2012New},
\[  e^{\beta} =2.63815853035(2).\]
The one exception is the honeycomb lattice, for which
it was predicted in 
\cite{Nienhuis1982Exact} and  proved in \cite{DuminilCopin2012Connective} that ${ e^{\beta}}=\sqrt{2+\sqrt{2}}$.  
We will set 
\[           q_n = e^{-\beta n} \, c_n , \]
and note that $q_n^{1/n} \rightarrow 1$.  We also set
\begin{equation}\label{eqn:pn_defn}
p_n = \frac{q_{2n}}{(q_n)^2},
\end{equation}
which is the probability that two SAWs of length $n$ can be concatenated to give
a SAW of length $2n$.  We will be primarily interested in
the square lattice $\Z^2$, but we first review what is known
about $\Z^d$.

In the trivial case $d=1$, $c_n = q_n = 2, p_n = 1/2$ for
all $n \geq 1$.  If $d \geq 5$, it
is known~\cite{Hara1992Selfavoiding}   that $p_n \rightarrow
p_\infty > 0$; this is suggested by (but harder to prove than) 
the fact that two independent ordinary walks in five
or more dimensions have a positive probability of no
intersection.  
Little is known rigorously about $q_n$
  for dimensions $2,3,4$. It is widely believed that 
  \[    q_n \sim   \begin{cases}A\, n^{\gamma_d - 1}  & d=2,3  \\
  A\, (\log n)^{1/4} & d = 4.\end{cases}\]
  Here $ A $ is a lattice dependent constant\footnote{Thoughout
  this paper we will use $A$ for a lattice dependent constant;
  however, its value will change from line to line. 
  However, when we write $A_1,A_2,\ldots$, these values
  will stay fixed.  In all cases, the values of these
  ``nonuniversal'' constants are not important to us.}, 
  $\gamma_2 = 43/32$ (this was first predicted
  by Nienhuis~\cite{Nienhuis1982Exact}
  and also follows from SLE computations, see Section \ref{scalingsec});
  $\gamma_3=1.156957 \pm 0.000009$ \cite{Clisby_gamma}.
For $d=4$, see 
 \cite{Bauerschmidt2014Logarithmic} for closely related results on
 a similar model.
In particular $p_n \asymp n^{-11/32}$ for $d=2$.   For
ordinary walks without self-avoidance,
 it is known \cite{cutpoints,LSWinter} that $p_n \asymp n^{-5/8}$ where $p_n$
 here is defined to be the probability that the 
 set of points visited by   two simple
 random walks of $n$ steps starting at adjacent points are disjoint.
(In our notation, $a_n \sim b_n$ means that $a_n/b_n
\to 1$ and $a_n \asymp b_n$ means that there exists $c$ such
that $a_n/c \leq b_n \leq c \, a_n$ for all $n$ sufficiently
large.)

The average size of an $n$-step SAW is often measured by the \emph{mean squared end-to-end distance} $\E_n\left[ |\omega_n|^2 \right]$, 
where $\E_n$ denotes expectation
with respect to   
  the uniform probability measure on walks of length $n$ starting at the origin. 
  It is widely believed that
  \[     \E_n \left[ |\omega_n|^2 \right] \sim \left\{
  \begin{array}{ll} A \, n^{2\nu} , & d \neq 4 \\
                   A \, n \, (\log n)^{1/4} , & d = 4  \end{array}
                    \right. , \]
 where $\nu  $ is a dimension dependent exponent. 
In the trivial case $d=1$, we have $\nu = 1$ and for
$d \geq 5$  this has been proved~\cite{Hara1992Selfavoiding} with $\nu = 1/2$. 
 In two dimensions it is predicted
  (see~\cite{Nienhuis1982Exact} and the SLE derivation below)
 that $\nu=3/4$, while in three dimensions the best current estimate is $\nu =  0.587597 \pm 0.000007$~\cite{Clisby_nu}.

For the remainder of this paper, we set $d=2$ and write
$\Z^2 = \Z \times i \Z$. 
 The theoretical (although at this point not mathematically rigorous) understanding
 of two-dimensional SAWs has been deepened by considering the implications of 
 assuming that the walks have a conformally invariant scaling limit.  Rather than
 working with walks of a fixed number of steps, it is more useful to consider the
 measure on SAWs of arbitrary length that gives measure $e^{-\beta n}$ to each
 SAW of $n$ steps, with any starting vertex.
 We will write $Q$ for this measure, that is, for any set $V$
 of SAWs,
\begin{equation*}
Q(V) = \sum_{n=0}^\infty  e^{-\beta n} \, c_n(V),
\end{equation*}
 where $c_n(V)$ is the number of walks of length $n$ in the set $V$.
In particular, if $V_n$ is the set of SAWs of length $n$ starting
at the origin, $Q(V_n) = q_n$.

 For any lattice spacing $1/N$, we can write $Q_{N}$ for the measure $Q$
 viewed as a measure on scaled SAWs on the lattice $N^{-1} \, \Z^2$.
 Then (see Section \ref{scalingsec} for more detail)
  if it is assumed that the measures $Q_N$ have a limiting
 measure that is conformally invariant, then this limit must be a version
 of   {\em Schramm-Loewner evolution   with parameter $8/3$}
 which we denote by \SLE.
 (There is a one-parameter family of  SLE  curves, but there is only
 one value $8/3$ that gives a measure on simple curves satisfying
 a ``restriction property'' which would have to hold for a scaling limit
 of SAWs.  For this reason we will only discuss \SLE~in this
 paper.  The values for exponents below such as $b,\tilde b$ are particular
 to the parameter $8/3$.)  The analogues of the
 exponents $\nu$ and $\gamma$ can be computed for \SLE~and give
 $\nu = 3/4, \gamma = 43/32$.  These are mathematically rigorous results
 about \SLE~although it is still only a conjecture (with strong 
 theoretical and numerical
 evidence) that it is the limit of the measure $Q$.

In this article we consider the enumeration of several different types of SAWs constrained to lie in the upper half-space of the square lattice, weighted with a Boltzmann weight corresponding to the maximum distance above the boundary of the half-space reached by the walks. When this weight is less than one, walks which step far away from the boundary are penalised, and one can thus view this system as a model of polymers at an impenetrable surface, subject to a force compressing them against the surface.

\section{The model}\label{sec:model}

Let $\saws$ denote the set of SAWs in $\Z^2 = \Z \times i \Z$ starting at the origin
 and let $\saws^+$ be the set of half-plane walks in $\saws$, that is, walks that
stay in
\[   \mathbb H := \{ x+iy \in \Z^2: y \geq 0 \} . \]
We let $\saws_n,\saws_n^+$ denote the set of such walks of length $n$. 

If $\omega = [0,\omega_1,\ldots,\omega_n] \in \saws^+$ we write 
 \[  h(\omega) = \max\{\Im(\omega_j):j=0,1,\ldots,n\} , \;\;\;\;
     y(\omega) = \Im(\omega_n) . \]
      We call
a walk $\omega = [\omega_0,\ldots,\omega_n] \in \saws^+_n$ a {\em bridge} 
if $h(\omega) = y(\omega)$.
 The utility of bridges lies in the fact that they can be freely concatenated to form larger bridges (with the addition of an extra step between).
      We write $\bridges_n$ for the set of bridges of length $n$ and $\bridges$ for the set of all bridges.
 We write
$\langle \cdot \rangle, \langle \cdot \rangle^+,\langle \cdot\rangle^b, \langle \cdot \rangle_n,
\langle \cdot \rangle_n^+,\langle \cdot\rangle^b_n$ for integrals (in fact, sums) with respect to the measure $Q$ restricted
to the sets $\saws, \saws^+,\bridges,\saws_n, \saws^+_n,\bridges_n,$
  respectively.  
For example,
\[\langle 1 \rangle_n = \sum_{\omega\in\saws_n} 1\cdot Q(\omega) = c_n e^{-\beta n} = q_n.\]
  Note
that $\langle 1 \rangle = \langle 1 \rangle^+ = \langle 1 \rangle^b = \infty$, as the corresponding generating functions diverge at $\beta$.

A walk $\omega = [\omega_0\ldots,\omega_{n-1}]\in\saws_{n-1}$ is called a \emph{polygon} if $|\omega_{n-1} -\omega_0| = 1$, that is, if the first and last vertices of $\omega$ are adjacent. An edge can be added between $\omega_0$ and $\omega_{n-1}$ to form a simple closed loop. We say that $\omega$ has length $n$. Let $\polygons$ be the set of all polygons and $\polygons_n$ the set of polygons of length $n$, with $\polygons^+$ and $\polygons_n^+$ the analogous sets restricted to those polygons staying in $\mathbb H$. Then let $\langle\cdot\rangle^p, \langle\cdot\rangle^p_n, \langle\cdot\rangle^{p+}, \langle\cdot\rangle^{p+}_n$ be the integrals with respect to the measure $Q$ restricted to $\polygons, \polygons_n, \polygons^+, \polygons_n^+$ respectively.

A polymer model which has been considered in the past~\cite{Guttmann2014Pulling, vanRensburg2013Adsorbed, Krawczyk2005Pulling, Skvortsov2012Mechanical} is that of polymers terminally attached to an impenetrable surface, with an external agent exerting a force on the non-attached end of the polymer, in a direction perpendicular to the surface. This reflects real-world experiments where polymers are pulled using optical tweezers or atomic force microscopy~\cite{Zhang2003}. This can be modelled using half-plane SAWs with a  Boltzmann weight  associated with the height of the end of the walks above the surface. The partition function of such a model is
\[\Cend_n(u) =  \langle  e^{-u y(\omega)} \rangle^+_n .\]
  One can interpret   $u$ as the reduced pulling force; when $u < 0$ the force is pulling up away from the surface, and when $u > 0$ the force is pulling down towards the surface. 
  As we can see below,  it is predicted that
   for fixed $y$, as $n \rightarrow \infty$,
 \[ \langle  \mathbbm 1_{\{y(\omega) = y\}} \rangle^+_n \asymp   
       n^{-19/16} \,  y^{25/28}, \]
 and hence the integral is dominated by the small values of $y$
which gives the following.
\begin{pred}\label{pred:endpoint}
For each $u > 0$, there exists a constant $A = A(u)$ such that
\[  \Cend_n(u) \sim A \, n^{-19/16}.\]
\end{pred}

It has recently been shown~\cite{Beaton2015} that this model undergoes a phase transition at $u=0$; that is, the free energy 
\[\kappa(u) := \lim_{n\to\infty}\frac1n \log C^\text{end}_n(u) + \beta\]
is non-analytic at $u=0$, being equal to $\beta$ for $u\geq0$ but strictly greater than $\beta$ for $u<0$. This is the transition between the \emph{free} and \emph{ballistic} phases of self-avoiding walks.

We will consider three similar quantities
\begin{align}
\Cmax_n(u)  &=  \langle  e^{-u h(\omega)} \rangle^+_n ,\label{eqn:qoi_walks}\\
B_n(u) &=  \langle  e^{-u h(\omega)} \rangle^b_n  =  \langle  e^{-u y(\omega)} \rangle^b_n,\label{eqn:qoi_bridges}\\
\Pmax_n(u) &= \langle e^{-u h(\omega)}\rangle^{p+}_n \label{eqn:qoi_polygons}
\end{align}
in the $u > 0$ regime.  We will use SLE to give predictions for the asymptotic
values and then will analyse exact   enumerations  to test the predictions. (For polygons, a numerical analysis has been conducted elsewhere \cite{gjw15}).
In this case the integral does not concentrate on the lowest order terms,
and there will be a stretched exponential decay.  We state the prediction
now.  We note that the constants $\lambda_1$ and $\lambda_2$ do not depend on $u$.

\begin{pred}\label{pred:main}
There exists $\lambda_1,\lambda_2 > 0$ such that for each $u > 0$, there exist $u$-dependent constants $A^+_u, A^b_u, A^{p+}_u$ such that
\begin{align*}
\Cmax_n(u) &\sim A^+_u \, n^{3/16} \, \exp\{-\lambda_1 u^{4/7} \, n^{3/7}\},\\
B_n(u) & \sim A^b_u \, n^{-13/28} \, \exp\{-\lambda_1 u^{4/7} \, n^{3/7}\},\\
\Pmax_n(u) &\sim A^{p+}_u\, n^{-11/7}\, \exp\{-\lambda_2 u^{4/7}\,n^{3/7}\}.
\end{align*}
\end{pred}

\section{A simpler model: ordinary random walks}

Here we consider an analogous problem for simple random walks
that can be solved rigorously.
Let $\langle \cdot \rangle_n^{+,{\rm simp}}$ denote expectations
with respect to simple random walks of length $n$ started at the
origin restricted to stay in the upper half-plane $\mathbb H$.  In other
words,
\[    \langle Y \rangle_n^{+,{\rm simp}}  =
   4^{-n} \sum_{\eta}  Y(\eta) \]
   where the sum is over all nearest neighbour (not necessarily
   self-avoiding) paths of $n$ steps
   starting at the origin with
   $\eta \subset \mathbb H$.  We write $h(\eta)$ for the maximal
  imaginary component of $\eta$.  We will give the asymptotics
  of  $\langle e^{-u h(\omega)} \rangle_n^{+,{\rm simp}} $ for $u > 0$
   as $n
  \to \infty$.
 
Let $c(n,h; r,y)$ denote the number of simple
random walk paths of   $n$ steps
starting at height $r$; ending at height $y$; whose
height stays  
between $0$ and $h-2$ for all times. 
  Then by viewing the imaginary part as a one-dimensional
  random walk killed when it leaves the interval $[0,h-2]$, we can see that 
$4^{-n} \, c(n,h;r,y) = J^n_h(r,y)$ 
where   $\{J_h(r,y): 0  \leq r,y  \leq h-2\}$ is the symmetric
matrix
given by
\[    J_h(r,r) = \frac 12 , \;\;\;\;J_h(r,r\pm 1) = \frac 14 . \]
The eigenvalues and eigenfunctions of $J_h$ are well known, and can be computed using a Fourier series (sum) over $\{0,1,\ldots,h-2\}$, and simple trigonometric identities.  If
${\bf v}_j$ denotes the vector with components $\left[\sin\left(\frac{j(k+1)\pi}{h}\right)\right]_ {k=0,\ldots,h-2}$, 
then
\[       J_h \, {\bf v}_j = \lambda_j \, {\bf v}_j, \;\;\;\;
   \mbox{where  } \;\; \lambda_j = 
 \frac 12 + \frac 12 \cos\left(
 \frac{j \pi}{h}\right). \]
By diagonalizing $J$, we can see that
\[      J_h^n(r,y)= \frac{2}{h}
 \sum_{j=1}^{h-1} \left[\frac 12 + \frac 12 \cos\left(
 \frac{j \pi}{h}\right)\right]^n \, \sin\left[\frac {(r+1) j \pi }{h}
  \right] \, \sin \left[ \frac{(y+1) j\pi}{h}\right].\]
If $h,n \to\infty$ with $h^2 \ll n$, then the $j=1$ term
dominates the asymptotics and
\begin{eqnarray*}
 J_{h}^n(r,y) & \sim &  \frac{2}{h} \,\left[\frac 12 + \frac 12 \cos\left(
 \frac{  \pi}{h}\right)\right]^n\,
 \sin\left[\frac {(r+1)  \pi }{h}
  \right] \, \sin \left[ \frac{(y+1)  \pi}{h}\right]\\
 & \sim &  \frac{2}{h} \exp\left\{- \frac{n\pi^2}{4h^2} \right\}\,
\sin\left[\frac {(r+1)  \pi }{h}
  \right]  \, \sin \left[ \frac{(y+1) \pi}{h}\right].
  \end{eqnarray*}
  In particular,
  \[   J_{h}^{n}(0,y) \sim \frac{2\pi}{h^2} \,
   \exp\left\{- \frac{n\pi^2}{4h^2} \right\}  \,  \sin \left[ \frac{(y+1) \pi}{h}\right].\]

Let $F_n(h)$ denote the probability that a simple
random walk starting at the origin in $\Z^2$ up to time $n$
stays in $\mathbb H$  and has maximum imaginary
component less than or equal to $h-2$.  
Then, if  $1 \ll h \ll\sqrt n $,
\begin{equation}\label{eqn:Fnh_asymp_exp}
F_n(h) = \sum_{y=0}^{h-2}  J_{h}^n(0,y)  \sim \frac{4}{h} \, \exp\left\{- \frac{n\pi^2}{4h^2} \right\} .
\end{equation}
      Note that $F_n(h)-
F_n(h-1)$ gives the probability that the  
  maximal imaginary component is 
exactly $h-2$.
 While we used the exact solution to determine these asymptotics,
 one can give a short heuristic to explain why the answer should be of the
 form $ c \, h^{-1} \, \exp\{-u(n/h^2)\}$ for some $c,u$.
 In the first $h^2$ steps, the probability that a walk
 stays in the upper half plane is $O(h^{-1})$ by the
 ``gambler's ruin'' estimate for one-dimensional walks.  After this,
 every time the walk moves $h^2$ steps, the chance that its imaginary
 component leaves $[0,h-2]$ is strictly between $0$ and $1$.  If
 we call this probability $1-e^{-a}$, then the probability that
 the walk stays in for $n$ steps should look like $[e^{-a}]^{n/h^2}$.  As $h \to \infty$, the exponential constant $a$ can be computed
 as an eigenvalue for Brownian motion.
 
We see that 
\begin{align}
 \langle e^{-u h} \rangle_n^{+,{\rm simp}}  & =
  \sum_{h=0}^\infty e^{-uh} \, [F_{n}(h+2) - F_n(h+1)]\notag\\
 & = e^{2u} \sum_{h=2}^\infty e^{-uh} \, [F_{n}(h) - F_n(h-1)]\notag\\
 & = e^{2u} (1-e^{-u}) \sum_{h=2}^\infty e^{-uh} \,  F_{n}(h) \notag\\
 & \sim 4 \, e^{2u} (1-e^{-u})\sum_{h=2}^\infty h^{-1}e^{-uh}\exp \left\{- \frac{n \pi^2}{4h^2}\right\},\notag\\
\intertext{where we have used~\eqref{eqn:Fnh_asymp_exp} for $h\ll\sqrt{n}$. For larger $h$, both $e^{-uh}F_n(h)$ and $h^{-1}e^{-uh}\exp\{-n\pi^2/4h^2\}$  decay quickly to 0. Then using \eqref{integralest} in the Appendix,}
&  \sim  4 \, e^{2u} (1- e^{-u  } ) \int_0^\infty x^{-1} \exp\left\{-\left(\frac{n\pi^2}{4x^2}+ux\right)\right\}dx \notag\\
   & \sim A_u \, n^{-1/6} \, \exp \left\{-\lambda_u \,n^{1/3} \right\}, \label{eqn:rw_asymps} 
\end{align}
where $A_u,\lambda_u$ can be derived from~\eqref{integralest}:
\begin{align*}
A_u &= 2^{8/3}\cdot3^{-1/2}\cdot\pi^{1/6}\cdot u^{-1/3}\cdot e^u(e^u-1),\\
\lambda_u &= 2^{-4/3}\cdot3\cdot\pi^{2/3}\cdot u^{2/3}.
\end{align*}

\section{Scaling limit}  \label{scalingsec}

The structure of this section is as follows. In Section~\ref{ssec:sle} we briefly review Schramm-Loewner evolution and demonstrate that, under the assumption that SLE$_{8/3}$ is the scaling limit of two-dimensional SAWs, it can be used to predict the asymptotic behaviour of the $Q$-measures of SAWs under various types of restrictions. In Section~\ref{ssec:saw_strip} we turn to the specific problem of compressed SAWs, which can be phrased in terms of walks restricted to a horizontal strip of fixed height, and together with Lemma~\ref{lem:integralest} in the Appendix derive the first two results in Prediction~\ref{pred:main}. In Section~\ref{ssec:polygons_strip} we consider polygons restricted to a horizontal strip. These require a slightly different approach, and we use results derived from restriction measures and large deviation theory to obtain the third part of Prediction~\ref{pred:main}.

\subsection{Schramm-Loewner evolution}\label{ssec:sle}

Here we review  predictions for the measure
$Q$  using the Schramm-Loewner evolution (SLE) in \cite{LSWsaw}
and extend them to incorporate information
about  the   natural parametrization
\cite{LShef,Content}. We start by explaining precisely
what is meant by the scaling limit of SAWs.
The starting point 
 is  to assume the existence of the mean-square displacement
exponent $\nu$ and two scaling exponents $b$ and $\tilde b$
that we discuss below. No
\emph{a priori} assumptions about the values of $\nu,b,
\tilde b$ are made.  It is convenient to view
a SAW $\omega = [\omega_0,\ldots,\omega_n]$
 as a continuous time process $\omega:
[0,n] \rightarrow \Z^2$ with $\omega(k) = \omega_k$ and otherwise
defined by linear interpolation.
The exponent $\nu$ is expressed by stating that 
  the fractal dimension of the
paths in the scaling limit is $1/\nu$. 
For every
lattice spacing  $1/N$ and $c > 0$, and every SAW $\omega = [\omega_0,\omega_1,\ldots,
\omega_n]$, we define the scaled walk $\omega^N$ on the lattice $N^{-1} \, \Z^2$
  by 
\begin{equation}  \label{natpar}
    \omega^N(t)  =   N^{-1} \, \omega(tN^{\nu}/c), \;\;\;\;
   0 \leq t \leq cn/ N^{\nu}. 
   \end{equation}
The measure $Q$ induces a measure on scaled paths in $N^{-1} \, \Z^2$.
Note that we do not change the measure except that we view it
as a measure on scaled paths. 
  If $\omega$ has length $n$,
then the time duration of the curve $\omega^N$ is $c n/ N^{\nu}$.
We discuss below what value to choose for $c$.

Let $\mathbb H$
denote the  upper half-plane as before.
If $z \in \Half  $ and $n$ is an integer, 
  let  $V_{0,nz}$ denote  
the set of SAWs in   $\Z^2$  starting at the origin and ending at $nz$.
(We are assuming for  notational   ease that
 $nz$ is  a lattice point; if not, choose the closest such
point.)
 Let $V_{0,nz}^+$ denote the set of half-plane SAWs
 in $V_{0,nz}$. The various types of restricted SAWs considered in this section can all be regarded in terms of $V_{0,nz}$ and $V^+_{0,nz}$, and so the limiting behaviour of the $Q$-measures of these sets is central to our predictions.

\begin{itemize}
\item {\bf Assumption}:  There exists a boundary scaling
exponent $b$, an interior scaling exponent
$\tilde b$, and  continuous functions $\psi,\psi_+$ on $
\Half $ 
such that as $n \rightarrow \infty$
\begin{equation}  \label{as1}
              Q[V_{0,nz}] \sim  
 \psi(z) \, n^{-2 \tilde b}. 
 \end{equation}
 \begin{equation}  \label{as2}
          Q[V_{0,nz}^+] \sim  \psi_+(z) \, n^{-b}
  \, n^{-\tilde b}. 
  \end{equation}
 \end{itemize}
 The scaling exponents $b,\tilde b$,
 come from the power law behavior in the total mass
 (partition function) of the measure $Q$.
 
   The scaling limit of the SAW
 is a collection of measures $\{\mu(0,z), \mu_\Half(0,z)\}$
 that are the limit of the measures
 $        n^{2 \tilde b} \, Q$ and $n^{b +\tilde b} Q$
 restricted to the sets $V_{0,nz}, V_{0,nz}^+$,
 respectively.  
 These are finite measures on simple curves from
 $0$ to $z$ in $\C$ or $\Half$, respectively.  They
 have total mass (normalized partition function) 
 $\psi(z),\psi_+(z)$ and one can get probability measures
 $\{\mu^\#(0,z), \mu_\Half^\#(0,z)\}$ by normalization.
 
 More generally, take $z,w$ to lie in the interior of $\Half$ (i.e. not on the boundary of the half-plane), and define $V_{nz,nw}$ to be the set of SAWs in $\mathbb Z^2$ starting at $nz$ and ending at $nw$. Then the measure $\mu(z,w)$ is the limit of the measure $n^{2\tilde b}Q$ restricted to the set $V_{nz,nw}$. Then if $D$ is an open subset of $\C$, 
 and $z,w$ are distinct interior points, we can define
 $\mu_D(z,w)$ to be the measure $\mu(z,w)$ restricted
 to curves that stay in $D$.  Equivalently, it is the
 scaling limit as above of SAWs from $nz$ to $nw$ that
 stay in $nD$.  This collection of measures satisfies
 the {\em restriction property}: if $D \subset D'$
 and $z,w \in D$, then $\mu_{D}(z,w)$ is $\mu_{D'}(z,w)$
 restricted to curves that  stay  in $D$.   Similarly,
 if $z \in \partial D$ and $w \in D$ we can define the
 measure $\mu_{ D}(z,w)$ provided that the boundary
 is sufficiently smooth at $z$ (there are lattice issues
 involved if the boundary of $D$ is not parallel with
 a coordinate axis (see \cite{LKen}), but we will not worry about this
 here.)
      
In addition to the restriction property, we will assume that the measures $\mu_D(z,w)$ satisfy two further criteria.
First, we  assume that the limit is in some sense
 conformally invariant:
 
 \begin{itemize}
 
 \item  The probability measures $\mu_D^\#(z,w)$ considered
 as a measure on curves modulo reparametrization\footnote{That is, if two curves $\gamma:[0,1]\to\mathbb C$ and $\tilde\gamma:[0,T] \to \mathbb C$ trace out the same path (but over different times), they are viewed as the same curve.}
 are conformally invariant.
 
 \end{itemize}
 
Given this, one can define the 
probability measures $\mu_D^\#(z,w)$ even for boundary
points at which the boundary is not smooth.  This allows
one to write down another property that the scaling limit
of SAW should have. It can be considered a property
of curves up to reparametrization.

\begin{itemize}

\item  Domain Markov property.  Suppose an initial
segment $\omega[0,t]$ is observed of a curve from $\mu_D^\#(z,w)$.
Then the distribution of the remainder of the curve
is that of $\mu_{D'}^\#(\omega(t),z)$ 
where $D'$ is the slit domain $D \setminus \omega[0,t]$. 

\end{itemize}
 
The three properties: conformal invariance, domain Markov property,
and restriction property characterize the scaling limit in that
there is only one family of measures on simple curves that
satisify all three.  The measure is called \SLE,
 the {\em Schramm-Loewner
evolution with parameter $8/3$}.  (There are SLE measures
with other parameters but they do not satisfy the restriction
property which a scaling limit of SAW would necessarily
satisfy.)

The critical exponents for SAW can be deduced theoretically
(but not at the moment mathematically rigorously) from 
mathematically rigorous theorems about 
  \SLE.  First, the fractal dimension of \SLE~paths is $4/3$.
  This was first proved as a statement about the Hausdorff
dimension, but for our purposes, it is more useful to
think of it in terms of the {\em $4/3$-Minkowski content}.
If $X$ is a compact subset of $\C$, then the $4/3$-dimensional
Minkowski content of $X$ is defined by
\[      \Cont(X) = \Cont_{4/3}(X) =
   \lim_{\epsilon \downarrow 0} \epsilon^{-2/3}
    \, {\rm Area}\{z: \dist(z,X) \leq \epsilon \}, \]
    provided that the limit exists.   If $\omega:[0,T] \rightarrow
    \C$ is a curve, then we can view $\Cont(\omega[0,t]) $
    as the ``$4/3$-dimensional length'' of $\omega[0,t]$.
 It has been shown that for the \SLE \,measures, the function
 $t \mapsto \Cont(\omega[0,t])$ is continuous and strictly
 increasing. For this reason we can parametrize our \SLE \,paths
  so that at each time $t$, $\Cont(\omega[0,t])=t$.  This
 is called the {\em natural parametrization} and is the
 $4/3$-dimensional analogue of parametrization by arc length.
 For the remainder we will assume that we have parametrized
 our curves in this way.  We conjecture that we can choose
 $c$ in \eqref{natpar} so that the curves in the scaling limit
 have the natural parametrization, and for the remainder
 we assume we have chosen this $c$.

 If $\omega:[0,T_\omega] \rightarrow \C$ is a curve
 taking values in a domain $D$ and $f:D \rightarrow f(D)$
 is a conformal transformation, then we define the image
 curve $f \circ \omega$ to be the image with the parametrization
 adjusted appropriately.  To be precise, the time
 to traverse $f(\omega[0,t])$ is 
\[    \int_0^{t} |f'(\omega(s))|^{4/3}
 \,ds . \]
 For example, if $f(z) = rz$ with $r > 0$, then the total time
 to traverse the curve is multiplied by $r^{4/3}$; this is
 the scaling property of a ``$4/3$-dimensional length''.
 If $\mu$ is a measure on curves in $D$, then we define
 the measure $f \circ \mu$ by
 \[    f\circ \mu(V) = \mu\{\omega: f \circ \omega \in V \}. \]
 We can now define the \SLE \,measures $\{\mu_D(z,w)\}$
 as a family of measure with the following properties. Here
 we can choose $z,w \in D$ (whole-plane \SLE); $z \in \partial
 D, w \in D$ (radial \SLE) or $z,w \in \partial D$ (chordal
 \SLE).  In the case of boundary points, we assume that $\partial
 D$ is locally analytic around the points.  The conformal
 covariance rule is
 \[     f \circ \mu_D(z,w) = |f'(z)|^{b_z} \, |f'(z)|^{b_w} \,
   \mu_{f(D)}(f(z),f(w)) ,  \]
   where $b_z$ and $b_w$ take values  $b$ or $\tilde b$, depending on
   whether $z$ and $w$ are boundary points or interior points.
 We can write
 \[          \mu_D(z,w) = \Phi_D(z,w) \, \mu_D^\#(z,w) , \]
 where $\Phi_D(z,w)$ is the total mass (normalized
 partition function) of the measure $\mu_D(z,w)$.  The
 covariance rule can then be stated as conformal invariance
 of the probability measures and a scaling rule for
 the partition functions,
 \[    \Phi_D(z,w)  =  |f'(z)|^{b_z} \, |f'(z)|^{b_w} \,
  \Phi_{f(D)}(f(z),f(w)) . \]
 Rigorously, we can show that 
 \[                  b = \frac 58, \;\;\; \tilde b = \frac{5}{48} , \]
 and the chordal, radial, and whole plane partition
 functions are each determined uniquely up to a mulitplicative
 constant by the scaling rule and the restriction property.
 For one choice of these constants, we conjecture that 
 we get the scaling limits of SAW as above.  We will assume
 this value of the constant (we do not know, even nonrigorously,
  this value). 
  
In order to relate the values of \SLE \,to SAWs of a given
number of steps, we must consider the time duration 
$T_\omega$ of the paths.  While it is known that the time
duration of \SLE \, paths can be given by the Minkowski
content, the next statement is currently only a conjecture about
\SLE.
\begin{itemize}
\item  The measures $\mu_D(z,w)$ can be written as
\[         \int_0^\infty \mu_D(z,w;t) \, dt , \]
where for each $t$,  $\mu_D(z,w;t)$ is a (strictly positive) finite
measure on curves $\omega$ from $z$
to $w$   with $T_\omega = t$.  We can write  
\[                 \mu_D(z,w;t) = \Phi_D(z,w;t)
 \, \mu_D^\#(z,w;t), \]
for a probability measure  $\mu_D^\#(z,w;t)$ 
with
$\Phi_D(z,w;t)$  continuous in $t$.
\end{itemize}
In other words,  with respect
to the probability measure $\mu_D^\#(z,w)$,
the random variable $T_\omega$ has a strictly positive
density given by $\Phi_D(z,w;t)/\Phi_D(z,w)$.

We now use this to compute (nonrigorously)
the exponent $\gamma$ for SAW in $\Z^2$.  
For the remainder of the paper, if we establish that two sequences $a_n$ and $b_n$ satisfy $a_n\asymp b_n$, we will often conclude
by conjecturing that there exists $c$ such that $a_n \sim
c \, b_n$.  We will not justify this last step, and, indeed,
in most cases we are just assuming that things are ``nice''.

Let $m_n$ denote the 
$Q$ measure of the
set  of SAWs starting at the origin
 of length  between $n^{4/3}$ and $(2n)^{4/3}$.
Then $m_n \asymp  n^{4/3} \, q_{n^{4/3}}$. 
Typically, these walks should go distance of order
$n$ and hence, we would expect that
the $Q$ measure of the
set  of SAWs starting at the origin
 of length between $n^{4/3}$ and $(2n)^{4/3}$ whose
 endpoint is distance between $n$ and $2n$ 
 is also comparable to $m_n$.  There are of order $n^2$
possibilities for this
endpoint, and the number for each of the possible
 endpoints should be comparable, so we would expect
 that the number of SAWs starting at $0$ ending
 at $n$ of  between $n^{4/3}$ and $(2n)^{4/3}$ steps
 is comparable to $m_n/n^2$.  By the scaling assumption,
 this is comparable to $n^{-2 \tilde b} = n^{-5/24}$ .
 Hence $m_n \asymp n^2 \, n^{-5/24} = n^{43/24}
  = [n^{4/3}] \, [n^{4/3}]^{11/32}.$ This
  gives the prediction $q_{n^{4/3}} \asymp [n^{4/3}]^{11/32}$,
\begin{pred}\label{pred:SAWs_asymp}
There exists $A$ such that
\[\langle 1 \rangle_n = q_n  \sim A \, n^{11/32}. \]
\end{pred}

The argument to compute  $m_n^+$,
 the $Q$ measure of SAWs starting
 at the origin of length $n^{4/3}$ to  $2n^{4/3}$ 
  restricted to the half-plane, 
  is similar.  The only difference is that
we now have one boundary point.  Hence, instead of
$n^{-2 \tilde b}$ we have $n^{-(b + \tilde b)}
  = n^{-35/48}$, and  $m_n^+ \asymp n^2 \, n^{-35/48}
    = n^{61/48} = [n^{4/3}] \, [n^{4/3}]^{-3/64}. $
\begin{pred}\label{pred:halfplane_asymp}
There exists $A^+$ such that
\[\langle 1 \rangle^+_n \sim A^+ \, n^{-3/64}. \]
\end{pred}

Finally, to compute $m_n^b$, the $Q$ measure
of bridges starting at the origin of length $n^{4/3}$ to  $2n^{4/3}$ 
  restricted to the half-plane, 
we now have two boundary points.  Hence,
$m_n^b\asymp n^2 \, n^{-5/4} = [n^{4/3}] \,[n^{4/3}]^{-7/16} $.
\begin{pred}\label{pred:bridge_asymp}
There exists $A^b$ such that
\[\langle 1 \rangle^b_n \sim A^b \, n^{-7/16}. \]
\end{pred}

Predictions~\ref{pred:SAWs_asymp} and~\ref{pred:halfplane_asymp} imply that the probability that
a SAW of length $n$ is a half-plane walk is comparable
to $n^{-25/64}$,  and Predictions~\ref{pred:halfplane_asymp} and~\ref{pred:bridge_asymp}
imply  that the
probability that a half-plane SAW of length $n$ is a bridge is comparable
to $n^{- 25/64}$. 
(It is not surprising that these are the same -- in the first case, we take SAW and insist that no vertex lie below the starting point; in the second case, we take a half-plane SAW and insist that no vertex lie above its final point. Symmetry suggests that the fraction of walks we keep ought to be roughly the same in both cases.)
Hence the probability that a SAW of length $n$ is
a bridge is comparable to $n^{-25/32}$. 

We finish by mentioning two other well-known conjectures.  Let $\Phi(\omega)$ be the indicator function
that the endpoints of $\omega$ are distance one apart.  Then observe that
$\langle 1\rangle^p_{2n} = \langle \Phi \rangle_{2n-1}$ and $\langle 1 \rangle^{p+}_{2n} = \langle\Phi\rangle^+_{2n-1}$.
\begin{pred}\label{pred:polygon_asymp}
There exist $A, A'$ such that
\[\langle \Phi \rangle_{2n-1} = \langle 1\rangle^p_{2n}\sim  A \, n^{-3/2}, \qquad \langle \Phi \rangle_{2n-1}^+ = \langle 1\rangle^{p+}_{2n} \sim A' \, n^{-5/2}.\]
\end{pred}

For the first relation, we note that the measure of
SAWs of $2n$ steps whose middle vertex is the origin
is comparable to $n^{11/32}$. The endpoints
of these walks are distance $n^{3/4}$ away and hence 
there are of order $n^{3/2}$ possibilities for 
each of the endpoints.  We therefore get a factor
of $n^{-3/2}$ which represents the probability that
 the endpoints of
the $2n$-step walk agree.  We also require the absence of any intersection between the first and second halves of
the walk near the initial/terminal point.  This gives
another factor that is comparable to $n^{-11/32}.$
Therefore, the measure should be comparable to
\begin{equation}  \label{dec23.1}
  n^{11/32} \, n^{-3/2} \, n^{-11/32} = n^{-3/2}.
  \end{equation}
One may note that the exponent $\gamma $ cancels in
this calculation --- one did not need to know its value.
This shows that SAPs can be easier to analyze than
SAWs since there are no endpoints.  
We will use the basic principle of \eqref{dec23.1}
several times below, so we state it.  This idea
extends to other dimensions, so we state it a little
more generally.
\begin{itemize}
\item  The probability that the endpoints of a SAW
of length $2n-1$ are distance one apart is comparable
to $p_{2n} \, n^{-d \nu}$, where $p_{2n}$ is as defined in~\eqref{eqn:pn_defn}.  In particular, if $d=2$
it is comparable to $ n^{-11/32} \,n^{-3/2} =  n^{-59/32}.$
\end{itemize}

For the second
relation in the prediction, we start with a SAP and consider the various
loops one gets by translating the root around.  
Typically there are only $O(1)$ such vertices with
minimum imaginary part, and hence the probability
that a SAP of length $2n$ is a
  half-plane SAP is comparable to $n^{-1}$.

Now let $q_{n,h}^+ $ denote the measure of the set
of SAWs of length $n$ that stay in the
upper half-plane and  that start at $ih$.  If $h=0$,
then $q_{n,0}^+ = q_n^+$ and if $h \geq n^{3/4}$, then
$q_{n,h}^+ \asymp q_n$.  We will consider the case
$0 < h \ll n^{3/4}$, and define $K$ by $n = K
 \, h^{4/3}$.  Then a SAW of $n$ steps can
 be viewed on a ``mesoscopic'' scale as  a SAW of
 $K$ steps where each of the steps is a little SAW of
 $h^{4/3}$ steps. Using this perspective, we see that
 we would conjecture that the probability
 that a SAW of $K h^{4/3}$ steps starting at $ih$
  stays
 in the half-plane is comparable to $K^{-25/64}$.
 Hence we conjecture that $q_{n,h}^+$ is comparable
 to $q_n \, K^{-25/64}$ which is comparable to
 $q_n^+ \, h^{25/48}$.  To phrase this in our
 notation, let $\hat S_h(\omega)$ be the indicator
 function that $\Im(\omega_j) \geq -h$ for all $j$.
 
\begin{pred}\label{pred:above-h}
There exists $A$ such that if $1 \ll h \ll n^{3/4}$, then
\[\langle \hat S_h \rangle_n \sim A \, n^{-3/64} \, h^{25/48}.\]
\end{pred}

If $z \in \Half$, let $q_n(z), q_n^+(z)$ denote the measure
of SAWs and half-plane SAWS, respectively, of length $n$
starting at the origin ending at $z$.  If $z = x + iy$
with $x \asymp n^{3/4}, y \asymp n^{3/4}$, then we expect
that $q_n(z) \asymp n^{-3/2}\, q_n,
 q_n^+(z) \asymp n^{-3/2}\, q_n^+$.  Now suppose $x 
 \asymp n^{3/4}$ and $1 \ll y \ll n^{3/4}$.  Then
 we view a SAW from $0$ to $y$ as the concatentation
 of two walks of length $n/2$ --- a half-plane walk
 starting at the origin and (the reversal of) a
 SAW starting at $x+iy$ that stays in the half-plane.
 We therefore guess that the measure should be comparable to
 $n^{-3/64} \, [n^{-3/64} \, y^{25/48}] \,
       n^{-59/32} = n^{-31/16} \, y^{25/48}.$ 
       Here the $n^{-59/32} = n^{-3/2} \, n^{-11/32} $
       as in the principle above.
  For a fixed $y$, there are of order $n^{3/4}$ possible
  values of $x$. By summing over these we get the following
  prediction.
       
\begin{pred}\label{pred:halfplane_fixedheight}
There exists $A$ such that if $1 \ll h \ll n^{3/4}$, then
\[\langle \mathbbm 1_{\{y(\omega) = h\}} \rangle^+ \sim A \, n^{-19/16} \, h^{25/48}.\]
\end{pred}

\subsection{SAWs restricted to a strip}\label{ssec:saw_strip}

The predictions in the last subsubsection
were made  assuming that the  the maximal
imaginary displacement  of the bridge
or half-plane walk  was  typical  for the number
of steps.  We now consider the case where this maximal height is
much smaller. We will fix a height $h$ and consider
SAWs restricted to the infinite strips
\[ \strip_h = \{x +iy \in \Z \times i \Z: 0 \leq y \leq h \},\]
\[  \strip_h^* = \strip_h - \frac{h}{2} \, i
 = \left\{x +iy \in \Z \times i \Z: - \frac h 2 \leq y \leq \frac{h}{2} \right\}
 .\]
We will write $S_h,S^*_h$ for the indicator functions
that $\omega \in \strip_h,\strip_h^*$, respectively.
We will consider SAWs in $\strip_h$ of length
$n = K \,h^{4/3}$ where $K \gg 1$.  We write
any such SAW $\omega$ as a concatenation of $K$ SAWs of length $h^{4/3}$:
\begin{equation}  \label{dec20.2}
 \omega =  \omega_1 \oplus \cdots \oplus \omega_K. 
 \end{equation}
 Let us choose a lattice space of $1/h$ and write the scaled
 walks as
 \[   \omega^h = \omega_1^h \oplus \cdots \oplus \omega_K^h. \]
These walks live in the strip
\[   D = \{x+iy \in \C: 0  \leq  y  \leq 1 \}. \]
The walk $\omega^h$ can be viewed as a polymer consisting
of $K$ monomers $\omega_j^h$ each of length of order $1$.
On this scale $\omega$ is a {\em one}-dimensional (in some
sense, weakly) self-avoiding walk of $K$ steps.  Hence, we would
expect the measure of such walks to decay like $e^{-\beta_1 K}$
for some $\beta_1$.  Moreover, we expect no smaller order corrections
in $K$ (that is, the measure is asymptotic to $A_1 \, e^{-\beta_1K}$
for some $A_1$)  and for the typical walk to look like a straight line
of length $cK$. This is the basis for the following predictions.  We
will assume that $h \rightarrow \infty$ and $K_h \rightarrow \infty$
and write $n = K_h \, h^{4/3}$.

\begin{pred}\label{pred:strip_walks}
There exist $\beta_1$ and $A_1,A_2,A_3,A_4$ such that the following holds.
\[   \langle S_h^* \rangle_n  \sim A_1 \, \exp \left\{-\beta_1 n \, h^{-4/3}\right\} \,  h^{11/24},\]
\[    \langle S_h \rangle_n^+  \sim  A_2 \,  \exp \left\{-\beta_1 n \, h^{-4/3}\right\}  \,  h^{-1/16} , \]
\[  \langle \mathbbm 1_{\{h(\omega) = h\}} \rangle_n^+  =  \langle S_h \rangle_n^+  - \langle S_{h-1} \rangle_n^+  \sim  A_3\,  \exp \left\{-\beta_1 n \, h^{-4/3}\right\}  \, [ n \, h^{-7/3} ] \, h^{-1/16},\]
\[ \langle \mathbbm 1_{\{y(\omega) = h)\}} \rangle_n^b \sim  A_4 \, \exp \left\{-\beta_1 n \, h^{-4/3}\right\}  \,  h^{-19/12} .\]
\end{pred}

The first two relations follow from Predictions~\ref{pred:SAWs_asymp} and~\ref{pred:halfplane_asymp} respectively, where the factor of $e^{-\beta_1 K}$ accounts for the height restriction and we have used $n=K h^{4/3}$. The third is merely the $h$-derivative of the second. The fourth follows from Prediction~\ref{pred:bridge_asymp}.
We note that  $h^{-19/12} = h^{-7/12} \, h^{-1}$.
The extra factor of $h^{-1}$ comes from the fact that we are
specifying the exact height. 
Using \eqref{integralest}, we see that  as $n \to \infty$,
\begin{eqnarray*}
\sum_{h=1}^\infty   e^{-hu} \, 
\exp
\left\{-\beta_1 n \, h^{-4/3}\right\} &\sim & 
 \int_0^\infty   
\exp
\left\{-\left(\frac{\beta_1 n}{ x^{4/3}}
  + ux \right)\right\} \, dx\\
  & \sim & A \, u^{-5/7}  \,n^{3/14} \, \exp \left\{-\lambda_1
  \, u^{4/7}\, 
    n^{3/7} \right\},\ 
\end{eqnarray*}
where
\[    \lambda_1 =7 \cdot 3^{-3/7}\cdot4^{-4/7}\cdot\beta_1^{3/7}, \]
and $A $ does not depend on $n$ or $u$. 
 Similarly, using \eqref{integralest},
  we see that there exists $c_1,c_2$ (independent
  of $u$)  such that with the same value of $\lambda_1$,
 \begin{eqnarray*}
 \langle e^{-uh(\omega)} \rangle_n^+ & = & \sum_{h=1}^\infty
         e^{-uh} \, \langle \mathbbm 1_{\{h(\omega) = h)\}} \rangle^+_n\\
      & \sim & \sum_{h=1}^\infty e^{-uh}
       \, A_3\,  \exp
\left\{-\beta_1 n \, h^{-4/3}\right\}  \, [ n \, h^{-7/3}  
  ]\, h^{-1/16}\\
  & \sim & A_3 n \int_0^\infty   x^{-115/48} \, \exp\left\{-\left(
  \frac{\beta_1 n}{ x^{4/3}} +ux\right)\right\} \, dx \\
  & \sim & c_1 \, n^{3/16} \, u^{5/16} \, 
   \exp \left\{-\lambda_1
  \, u^{4/7}\, 
    n^{3/7} \right\},
\end{eqnarray*}
\begin{eqnarray*}
 \langle e^{-uh(\omega)} \rangle_n^b & = & \sum_{h=1}^\infty
         e^{-uh} \, \langle \mathbbm 1_{\{h(\omega) = h)\}} \rangle^b_n\\
      & \sim & \sum_{h=1}^\infty e^{-uh}
       \, A_4\,  \exp
\left\{-\beta_1 n \, h^{-4/3}\right\}  \, h^{-19/12}\\
& = & A_4\int_0^\infty x^{-19/12} \, 
 \exp\left\{-\left(
  \frac{\beta_1 n}{ x^{4/3}} +ux\right)\right\} \, dx  \\
     & =  & c_2 \, n^{-13/28} \,u^{-1/28}
     \, \exp\left\{-\lambda_1 u^{4/7}n^{3/7} \right\}
    .
\end{eqnarray*}
We have thus obtained the results of Prediction~\ref{pred:main}. 

\subsection{Polygons restricted to a strip}\label{ssec:polygons_strip}

We finish this section with some predictions for self-avoiding
polygons in a strip.

\begin{pred}\label{pred:strip_polygons}
There exist $\beta_2$ and $A_5,A_6$ such that the following holds.
\[   \langle S_h^* \, \Phi \rangle_{n-1} = \langle S_h^*\rangle^p_n \sim  A_5 \, \exp\left\{-\beta_2 \,n \,h^{-4/3}\right\} \,  h^{-17/6}, \]
\[   \langle  S_h \, \Phi \rangle_{n-1} = \langle S_h\rangle^p_n \sim  A_6 \, \exp\left\{-\beta_2 \,n \,h^{-4/3}\right\} \,  h^{-25/6}.\]
\end{pred}

We will take a slightly different approach here.  For
any $0 \leq y,\tilde y \leq h$ and integer $N >0$, let
${\mathcal K}_h(y,\tilde y,N;r)$ denote the set of SAPs
$\omega = [\omega_0,\ldots,\omega_{2n-1},\omega_{2n}
= \omega_0]$
with the following properties:
\[    \omega_0 = i y, \;\;\;\; \omega_{2n-1} = i(y-1) \]
\[        \omega_{r} = N + i \tilde y , \;\;\;\;
   \omega_{r+1} = N+i(\tilde y-1) , ,\]
   \[  1 \leq \Re[\omega_j] \leq N-1 ,\;\;\;
     j \neq 0,r,r+1,2n-1,2n.\]
     \[  \omega \subset \strip_h , \]
 and let
 \[   {\mathcal K}_h(y,\tilde y,N) = \bigcup_r {\mathcal K}(y,\tilde y,N;
 r), \;\;\;\; K_{h,N} = \bigcup_{y,\tilde y} {\mathcal K}(y,\tilde y,N) .\]
In other words,
$ {\mathcal K}_{h,N}$ is the set of polygons in $\strip_h$ whose
minimal real value is $0$; maximal real value is $N$; 
that have only one bond on each of 
the vertical lines $I_0 := \{\Re(z) = 0\}$
and $I_N:=\{\Re(z) = N\}$;  and are translated and
oriented  so that they start at the 
``higher'' point on $I_0$ and that the final
bond is the unique bond on  $I_0$.  Let $\tilde  {\mathcal K}_{h,N}$
denote the set of polygons in $\strip_h$ whose
minimal real value is $0$; maximal real value is $N$; 
and such that the initial vertex is on $I_0$.  Let $J_{h,N},
\tilde J_{h,N}$ denote the corresponding indicator functions.
In the set $\tilde  {\mathcal K}_{h,N}$, the SAP can have multiple
bonds on the extremal vertical lines $I_0,I_N$; however, we
predict that typically there are only a few such bonds.
In particular, we predict that 
 there exists $A$ such that as $h,N \rightarrow \infty$ with $N/h \rightarrow \infty$,
\[
\langle S_h^* \,\Phi\, \tilde  J_{h,N} \rangle \sim A \,
   \langle S_h^* \,\Phi\, J_{h,N} \rangle . \]

If $\omega \in  {\mathcal K}_h(y,\tilde y,N)$, then we can view $\omega$
as two SAWs: $\omega_1$ from $iy$ to $N + i \tilde y$ and
$\omega_2$ from $i(y-1)$ to $N+ i (\tilde y - 1)$.  The
polygon $\omega$ is obtained by concatenating $\omega_1$
and (the reversal of) $\omega_2$ adding the two extra
bonds to make this a polygon.  Any $\omega_1,\omega_2$
can be chosen provided that $\omega_1 \cap \omega_2
 = \emptyset$.  So we need to find the $Q$ measure of
 the set of pairs of such walks $(\omega_1,\omega_2)$ with
 $\omega_1 \cap \omega_2 = \emptyset$. The limiting measure
 is predicted to be the restriction measure with exponent
 $2$.  To be precise it is predicted~\cite{LSuniversality, LSWchordal}
that the measure is asymptotic to 
 \[         A \, h^{-4} \,  H_D(i(y/h), K + i(\tilde y/h))^{2} \]
 where $H_D$ denotes the boundary Poisson kernel (derivative
 of the Green's function), $K = N/h$ and
 \[   D = D_K = \{x+iy:  0 < x < K, 0 < y < 1 \} . \]
 The factor $h^{-4}$ should be viewed as a factor of $h^{-2}$
 for each boundary end of the SAP where $2$ is the exponent
 of the restriction measure.
 Using a conformal transformation we see that
 \[  H_D(i(y/h), K + i(\tilde y/h)) 
  \sim c \, \sin(y \pi/h) \, \sin(\tilde y \pi /h)  \,
         e^{-K\pi}.\]
   If we sum over the
 number of possible values of $(y,\tilde y)$, we
see that we get the prediction
   \[    Q({\mathcal K}_{h,N})  \asymp
   Q (\tilde {\mathcal K}_{h,N}) 
   \asymp 
  h^{-2} \, 
                     e^{-2K\pi } . \]
                      
 Now consider the set ${  \hat  {\mathcal K}_{h,N}}
 $ of self-avoiding polygons in $\mathcal S^{*}_h$ starting at 
  the origin such that the real displacement equals $N
   { = K \, h  }$.
 To each such polygon, we can find a polygon that visits the
 same points in the same order that starts at a point
 of minimal real part.  Conversely, if we have a walk in $
 {  \hat {\mathcal K}_{h,N}}$
 and choose a vertex in the middle, we can translate the walk 
 so that  vertex is the starting point.  Since the typical
 such polygon has length comparable to $K \, h^{4/3}$,
 there are  $O(Kh^{4/3})$
 choices for the vertex.  We therefore get the conjecture
 \[    Q({   \hat  {\mathcal K}_{h,N}}) \sim A \,K \, h^{4/3} \,
  h^{-2} \, 
                     e^{-2K \pi}  = A \, K \, h^{-2/3} \, 
                      e^{-2K \pi} . \]
                     
 This gives a prediction for a fixed real displacement, but we
 must convert it to a fixed number of steps.  Consider 
 a walk in $ \hat  {\mathcal K}_{h,N}  $ chosen from the
  probability measure given by $Q$ restricted to
  $ \hat  {\mathcal K}_{h,N}$, normalized to be a probability
  measure.  Let $T$ denote the number of steps of such
  a walk.  $T$ is the sum of $K$ random variables, each representing
  the number of steps in one of the $K$ squares of side length 
  $h$.  These random variables are roughly identically distributed
  and have short range correlations, so as an approximation we view
  $T$ as having the behavior of
  \[      h^{4/3} \, [X_1 + \cdots + X_K] , \]
  where $X_j$ are independent, identically distributed.  We write
  $\alpha$ for the mean.  Hence, we predict that the expectation of
  $T$ is $\alpha h^{4/3} K$ and the standard deviation of $T$
  is comparable to $h^{4/3} \, K^{1/2}$. 
   Even more precisely, standard arguments from large deviation theory lead us to expect that (at least for small $r$)
   there is a rate function $\rho(r)$ with $\rho(\alpha)
    = 0$  such that  the 
   probability that $T  =2n$ where $ n
    = h^{4/3} \,r \, K = h^{1/3} \, r
     \, N$ is  comparable to
    \[            \left[h^{ 4/3} \, K^{1/2}\right]^{-1}
       \, \exp\{- \rho(r) \, K\}. \]
    Hence the $Q$ measure of polygons in  $ \hat  {\mathcal K}_{h,N}  $  of $2n$ steps is comparable to
\[\left[h^{ 4/3} \, K^{1/2}\right]^{-1} \, \exp\{- \rho(r) \, K\} \, K \, h^{-2/3} \, e^{- 2 \pi K }
       \asymp  N^{1/2}\, h^{-5/2} \, \exp \left\{-\left[\rho\left(\frac{n}{h^{1/3}
        \, N} \right)  
        + 2 \pi\right] N \, h^{-1}
        \right\},
       \]
       and
\[ 
  \langle S_h^* \, \Phi \rangle_n  \asymp  
     h^{-5/2} \, \sum_{N=1}^\infty 
         N^{1/2}\,  \exp \left\{-\left[\rho\left(\frac{n}{h^{1/3}
        \, N}  \right)  
        + 2 \pi\right] N \, h^{-1}
        \right\}.\]
 As before, to analyze this sum we start by finding the $N$
 that minimizes
 \[   f(N) = \left[\rho\left(\frac{n}{h^{1/3}
        \, N} \right)  
        + 2 \pi\right] N = \left[\rho(r) + 2 \pi
      \right] N  .\]
  Since \[f'(N) = 
       \rho\left(\frac{n}{h^{1/3}
        \, N}\right)  + 2 \pi - \frac{n}{h^{1/3} N } \rho'
        \left(\frac{n}{h^{1/3}
        \, N} \right) = \rho(r) + 2 \pi
         - r \, \rho'(r) , \]
 we find $r_0$ with
 \[       r_0 \, \rho'(r_0) - \rho(r_0) = 2 \pi.\]
We assume that $\rho$ is smooth and we note that
it has a global minimum at $r = \alpha$ where
$\rho(\alpha) = \rho'(\alpha) = 0$.  Therefore,
we expect $r_0 > \alpha$, and if we set $\lambda =
\rho(r_0)$, 
\[      \rho(r_0 + \epsilon) = \lambda +
 \rho'(r_0) \, \epsilon +O(\epsilon^2) =
     \lambda  + \frac{2 \pi + \lambda}{r_0} \, \epsilon
    + O(\epsilon^2).\]
The terms in the sum are maximized at $N_n$ where
  $N_n h^{1/3}r_0  = n  $.
Moreover, for $k \ll N_n$,
\begin{align*}
  f(N_n + k) & =  \left[\rho\left(\frac{n}{h^{1/3}
        \, (N_n+k)} \right)  
        + 2 \pi\right]\, (N_n+k) \\
         &  =  \left[\rho\left(r_0 
         -r_0 \, \frac{k}{N_n} + 
       O\left(\frac{k^2}{ N_n^2} \right)\right)
        + 2 \pi\right]  \, (N_n+k) \\
    & = \left[\lambda  - \frac{(\lambda + 2\pi)k}{N_n}
        +  O\left(\frac{k^2}{ N_n^2} \right) + 2 \pi\right] \, (N_n+k)\\
        & = f(N_n) + O\left(\frac{k^2}{N_n}\right) . 
        \end{align*}
 The key fact about this calculation is that $f(N_n+k)$ is
 comparable to $f(N_n)$ for $k^2 \leq N_n$, that is,
  for $O(N_n^{1/2})$ values of $k$ and
 hence
 \[ \sum_{N=1}^\infty 
         N^{1/2}\,  \exp \left\{-\left[\rho\left(\frac{n}{h^{1/3}
        \, N}  \right)  
        + 2 \pi\right]\, N \, h^{-1}
        \right \} \asymp N_n \,  \exp \left\{-\beta_2 nh^{-4/3}
       \right\}
      \asymp \left(n/h^{1/3}\right)  \, \exp \left\{-\beta_2 n h^{-4/3}
       \right\},   \]
       and
 \[        \langle S_h^* \, \Phi \rangle_n  \asymp n \, h^{-17/6}
  \,  \exp \left\{-\beta_2 nh^{-4/3}
       \right\} \]
for some $\beta_2$ that we cannot determine explicitly.
 
   For the second part of Prediction~\ref{pred:strip_polygons}, we start by considering a polygon with
  $S_h^* \, \Phi = 1$ as two parts, $\omega^1$ and $\omega^2$.  There is a 
  ``lower'' SAW $\omega^1$
   of
  $2 h^{4/3}$ steps whose middle vertex is the origin, and then  
  $\omega^2$ is an ``upper'' $n - 2h^{4/3}$
  step SAW from one endpoint of $\omega^1$ to the other that does
  not intersect $\omega^1$.  Using the reasoning for polygons
  from before, we predict that the probability that $\omega^1$
  lies entirely in the upper half-plane is comparable to $h^{-4/3}$.
  Hence, we would expect
  \[   \langle S_h \, \Phi \rangle_n  \asymp h^{-4/3}
   \, \langle S_h^* \, \Phi \rangle_n .\]
   
Finally, we can associate the weight $e^{-uh}$ with polygons of height $h$, as we did previously with walks and bridges. To do so we must first examine polygons of height exactly $h$, rather than height at most $h$.

\begin{pred}
There exist constants $A_7$ and $\beta_2$ (taking the same value as in Prediction~\ref{pred:strip_polygons}) such that
\begin{equation}\label{eqn:polygons_height_exactly_h}
\langle S_h\,\Phi\,\mathbbm 1_{\{h(\omega)=h\}}\rangle_n = \langle \mathbbm 1_{\{h(\omega)=h\}}\rangle^{p+}_n = \langle S_h\rangle^p_n - \langle S_{h-1}\rangle^p_n \sim A_7\,n\,h^{-13/2}\exp\left\{-\beta_2\,n\,h^{-4/3}\right\}.
\end{equation}
\end{pred}
\noindent This is obtained by taking the $h$-derivative of the second expression in Prediction~\ref{pred:strip_polygons}.

The asymptotics of $\langle e^{-u h(\omega)}\rangle^{p+}_n = \Pmax_n(u)$ as given in Prediction~\ref{pred:main} are then obtained by multiplying~\eqref{eqn:polygons_height_exactly_h} by $e^{-uh}$, summing over $h$, and using~\eqref{integralest}.

\section{Numerical analysis}\label{sec:analysis}

We now wish to empirically test the validity of Prediction~\ref{pred:main}. That is, we wish to investigate the behavior of $C_n^\text{max}(u)$ and $B_n(u)$ for a variety of values of $u$, by generating and analyzing the sequences up to as large a value of $n$ as possible.

In full generality then, we will be testing the hypothesis that two one-parameterized
sequences of positive
numbers $r_n(u), \,\, s_n(u)$, depending on $u,$ have an asymptotic form
\[     r_n(u) \sim A_u \, \exp\left\{-\lambda u^{4/7}\, n^{3/7} \right\}
 \, n^g , \]
 \[  s_n(u) \sim A_u' \, \exp\left\{-\lambda u^{4/7}\, n^{3/7} \right\}
 \, n^{g'} , \]
respectively, where $\lambda, A_u, A_u',g,g'$ are unknown and $\lambda,g,g'$ are independent
 of $u$. (Naturally, we will wish to use $r_n(u) = C^\text{max}_n(u)$ and $s_n(u) = B_n(u)$.) This can also be written as
\[   \log r_n(u) = \log A_u - \lambda \, u^{4/7} \, n^{3/7}
  \, + g \, \log n +o(1) .\]
  We will make the stronger assumptions that the $o(1)$ is actually
  $O(n^{-4/7})$ and that a derivative form holds,
  \[  \log r_{n+2}(u) - \log r_n(u)
     = \left(\frac{6\lambda \, u^{4/7}}{7 n^{4/7}} 
      + \frac{2g}{n}\right) \, [1 + O(n^{-1})],\]
\[\frac{r_n(u)}{s_n(u)} = const\cdot n^{g-g'} + O(n^{-11/7}).\]
One starts by estimating $\lambda$ using the values of $\log r_n(u)$
for different values of $u$ and finding the slope. 

\subsection{Generation of data}

Let $b_{n,h}$ be the number of bridges of length $n$ spanning a strip of width $h.$ The generating function is 
\[B(z,u)=\sum_{n,h} b_{n,h} z^n e^{-uh} = \sum_n B_n(u) z^n e^{\beta n},\]
where $B_n(u)$ is as defined in~\eqref{eqn:qoi_bridges}. We have generated all coefficients in $B$ with $n \le 86$ on the square lattice. Thus $h \le 86$ also. 

If $c^\text{max}_{n,h}$ is the number of upper half-plane SAWs of length $n,$ originating at the origin, with maximum height $h$, the corresponding generating function is 
\[C^\text{max}(z,u)=\sum_{n,h} c^\text{max}_{n,h} z^n e^{-uh} = \sum_n C^\text{max}_n(u) z^n e^{\beta n},\]
where $C^\text{max}_n(u)$ is as defined in~\eqref{eqn:qoi_walks}. We have generated all coefficients in $C^\text{max}$ with $n \le 60$ on the square lattice.

The algorithm we use to enumerate SAWs and bridges on the square lattice builds on the 
pioneering work of Enting \cite{Enting80}  on  
self-avoiding polygons extended to walks   by Conway, Enting and 
Guttmann \cite{Conway93} with further enhancements over the years
by Jensen and others \cite{Jensen04,Clisby2012New,Jensen13}.    Below we shall only briefly
outline the basics of the algorithm and describe the changes made
for the particular problem studied in this work.

The first terms in the series for the SAW generating function can be calculated 
using transfer matrix techniques to count the number of SAWs in rectangles of
width $w$ and length $l$ vertices long. Any SAW spanning such a rectangle has length 
at least $w+l-2$. By adding the contributions from all rectangles of width 
$w \leq n+1$  and length $w \leq l \leq n-w+1$ the number of SAW  
is obtained correctly up to length $n$.

The generating function for rectangles with fixed width $h$ are calculated using 
transfer matrix (TM) techniques. The most efficient implementation of the TM algorithm 
generally involves  cutting the finite lattice with a line 
and moving this cut-line in such a way as to build up 
the lattice vertex by vertex.  Formally a SAW can be viewed as  a sub-graph of the square lattice
such that every vertex of the sub-graph has degree 0 or 2 (the vertex is empty or the SAW
passes through the vertex) excepts for two vertices of degree 1 (the start- or end-point of the SAW).
Thus if we draw a SAW and  then cut it by a line we observe that the partial SAW to the left and right of this 
line consists of a number of arcs connecting two edges on the intersection, 
and pieces which are connected to only one edge (we call these free ends). 
The other end of a free piece is  either the start- or  end-point of the SAW so there are 
at most two free ends. 

Each end of an arc is assigned one of two labels 
depending on whether it is the lower end or the upper end of an arc. Each 
configuration along the cut-line  can thus be represented by a set of 
edge states $\{\sigma_i\}$, where

\begin{equation}\label{eq:states}
\sigma_i  = \left\{ \begin{array}{rl}
0 &\;\;\; \mbox{empty edge},  \\ 
1 &\;\;\; \mbox{lower arc-end}, \\
2 &\;\;\; \mbox{upper arc-end}. \\
3 &\;\;\; \mbox{free end}. \\
\end{array} \right.
\end{equation}
\noindent
Since crossings aren't permitted this encoding uniquely describes 
which loop ends are connected.

The sum over all contributing graphs is calculated as the cut-line 
is moved through the lattice.  For each configuration of occupied or empty edges 
along the intersection we maintain a generating function $G_S(z)$ for partial walks 
with signature $S$, where $G_S(z)$ is a polynomial truncated at degree $n$.  
In a TM update each source  signature $S$ (before 
the cut-line is moved) gives rise to a  few new target signatures  $S'$ 
(after the move of the boundary line) and $n=0, 1$ or 2 new edges are inserted 
leading to the update  $G_{S'}(z)=G_{S'}(z)+z^nG_S(z)$. Once a signature $S$ 
has been processed it can be discarded.

Some changes to the algorithm described in \cite{Jensen04,Jensen13} are required in order to
enumerate the restricted SAWs and bridges. We used the recently developed version of the TM 
algorithm \cite{Clisby2012New,Jensen13} in which edge states describe how the set of occupied 
edges along the cut-line are to be connected to the right of the cut-line. i.e, how edges must 
connect as the cut-line is moved in the transfer direction from left to right.

\subsubsection{Further details for SAWs}

Grafting the SAW to the wall can be achieved by forcing one of the free ends  (the start-point)  to lie
on the bottom side of the rectangle. In enumerations of unrestricted SAWs one can use 
symmetry to restrict the TM calculations to rectangles with  $w\leq n/2+1$ and $l\geq w$ 
by counting contributions for rectangles with $l>w$ twice. 
The grafting of the start-point to the wall breaks the symmetry and we have to consider
all rectangles with $w\leq n+1$. The number of configurations one need consider
grows exponentially with $w$. Hence one wants to minimise the length of the cut-line. To achieve
this the TM calculation on the set of rectangles is broken into two sub-sets with
$l\geq w$ and $l<w$, respectively.  In the calculations for the  sub-set with $l<w$ the  
cut-line is chosen to be horizontal (rather than vertical) so it cuts across at most  $l+1$ edges.
Alternatively, one may view the calculation for the second sub-set as a TM algorithm for
SAWs with its start-point on the left-most border of the rectangle.

In order to measure the maximum height of the SAW all we need to do is extract this information
from the finite-lattice data set. From the calculation with  $l\geq w$, where one end-point must lie 
on the bottom of the rectangle, the maximum height is simply $w$.  Similarly for $l<w$, where one end-point 
must lie  on the left of the rectangle,  the maximum height is $l$. The final series is then obtained by combining
the results from the two cases summing over all possible sizes of the rectangles.

\subsubsection{Further details for bridges}

As for SAWs the calculation for  bridges is broken into two sub-sets corresponding to
bridges which span the rectangle from bottom-to-top when $l\geq w=h$ or from
left-to-right when $h=l<w$. In this case there are some further restrictions on the permissible 
types of configurations which makes the algorithm a little  more efficient. The most important
is that at all stages during the TM calculation a free end must have access to the 
boundaries of the rectangle otherwise the partially constructed SAW would
be unable to span the rectangle and hence could not lead to a bridge. So any configuration
where a free end is embedded inside arcs is forbidden.

Clearly concatenating two bridges of height $i$ and $j$ gives a bridge of height 
$i+j$ (we place the origin of the second walk on top of the end-point of the first walk). This means
that any bridge can be decomposed into {\em irreducible bridges}, i.e.,
bridges which cannot be decomposed further, and we use $a_n$ to
denote the number of $n$-step irreducible bridges. So 
the generating function $B(z,u)$ for bridges is simply related 
to the generating function for irreducible bridges $A(z,u)$

\begin{equation}\label{eq:A2B}
B(z,u)=\frac{A(z,u)}{1-A(z,u)}.
\end{equation}

This fact can be used to extend the series for bridges by the following simple
observation. If we calculate the number of bridges in rectangles up to
some maximal width $w_{\rm max}$ the series for bridges will be correct to
order $2w_{\rm max}+1$ but if we do the same for irreducible bridges 
the series will be correct to order $3w_{\rm max}+2$.  Thankfully the
number of irreducible bridges can easily be obtained from the number of bridges.
Consider the number of bridges $b_{n,h}$ and irreducible bridges $a_{n,h}$ 
of length $n$ and height  $h$ with associated generating functions $B_h(z)$ and $A_h(z)$. 
 Since a bridge is either irreducible or the concatenation of a bridge
with an irreducible bridge we get 

$$B_h(z) = A_h(z) + \sum_{k=1}^{h-1} A_{h-k}(z)B_k(z)$$
\noindent
and thus 

$$A_h(z) = B_h(z)-\sum_{k=1}^{h-1} A_{h-k}(z)B_k(z),$$
\noindent
which allows us to obtain all generating functions $A_h(z)$ recursively from 
$B_h(z)$ for $1 \leq h \leq w_{\rm max}$ and once these are known we can extend the
series  $B(z,y)$ by~(\ref{eq:A2B}) to order $n=3w_{\rm max}+2$ for any $1<h\leq n$.

We calculated series for $B_h(z)$ up to $h=28$ and  thus $n=86$ which on its own
would give a series for bridges to order 57 but extracting the data for irreducible bridges
and using~\eqref{eq:A2B} the series to order 86 was obtained. The calculation used a total
of some 25,000 CPU hours. The most demanding part was the calculation for
$w=28$ where we used 640 processor cores  with 3GB of memory used
per core.

\subsection{Analysis of data}
The sequences we now analyze are the coefficients of $B(z,u)$ and $C^\text{max}(z,u)$; that is, $\{B_n(u)e^{\beta n}\}_n$ and $\{C^\text{max}_n(u)e^{\beta n}\}_n$. For brevity, in the remainder of this paper we will use $b_n \equiv B_n(u)e^{\beta n}$ and $c^+_n\equiv C^\text{max}_n(u)e^{\beta n}$, and will use $a_n$ as a placeholder sequence which could stand for either.

As discussed in  \cite{G14a}, if the asymptotic form of a sequence $\{a_n\}$ is
\begin{equation} \label{eq:coef}
a_n \sim const \cdot e^{\beta n} \cdot \mu_1^{n^\sigma} \cdot n^{\gamma-1},
\end{equation}
then the ratio of successive coefficients is
\begin{equation} \label{eq:rn}
r_n = \frac{a_n}{a_{n-1}}=e^{\beta} \left (1 + \frac{\sigma \log {\mu_1}}{n^{1-\sigma}} +\frac{g}{n}+\frac{\sigma^2 \log^2 \mu_1}{n^{2-2\sigma}}+ O\left (\frac{1}{n^{2-\sigma}}\right ) \right ).
\end{equation}
So a ratio plot against $n^{\sigma-1}$ should be linear, with perhaps some low-$n$ curvature induced by the presence of the $O(1/n)$ term.

To illustrate this, we take the generating function for bridges with $u=\log 2$ as a representative value, not subject to cross-over effects present near $u=0$ and $u=\infty$. Plotting the ratios of pushed bridges against $1/n,$ as shown in Figure~\ref{fig:ratiosa} gives a plot displaying considerable curvature. Note that the limiting value of $e^{\beta}$ as $n \to \infty$ is given by the top left corner of the plot. Plotting the ratios against $1/n^{2/3},$ as shown in Figure~\ref{fig:ratiosb} gives a plot with reduced, but still significant, curvature. In Figures~\ref{fig:ratiosc} and~\ref{fig:ratiosd} we show ratio plots against $1/\sqrt{n}$ and against $1/n^{4/7}$ respectively. Both are visually linear, but more careful extrapolation reveals that plot~\ref{fig:ratiosc} extrapolates, as $n \to 0,$ to a value slightly greater than the numerical value of $e^{\beta},$ while plot~\ref{fig:ratiosd} extrapolates to a value indistinguishable from the numerical value of $e^{\beta}.$ Thus this simple sequence of ratio plots gives good evidence that the ratios are linear when plotted against $1/n^\theta,$ where $0.5 \le \theta \le 0.6,$ and is totally consistent with the heuristic expectation, $\theta=4/7 = 0.57142.. .$

\begin{figure}[t!]
\centering
\begin{subfigure}[b]{0.49\textwidth}
\includegraphics[width=\textwidth]{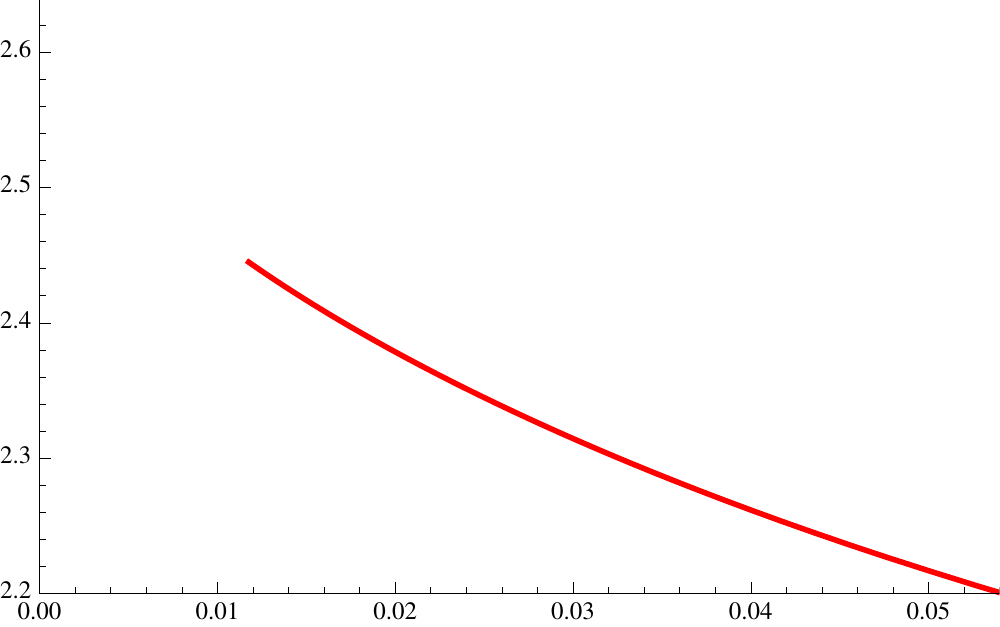}
\caption{Plot of ratios of coefficients against $\frac{1}{n}$.}
\label{fig:ratiosa}
\end{subfigure}\hfill
\begin{subfigure}[b]{0.49\textwidth}
\includegraphics[width=\textwidth]{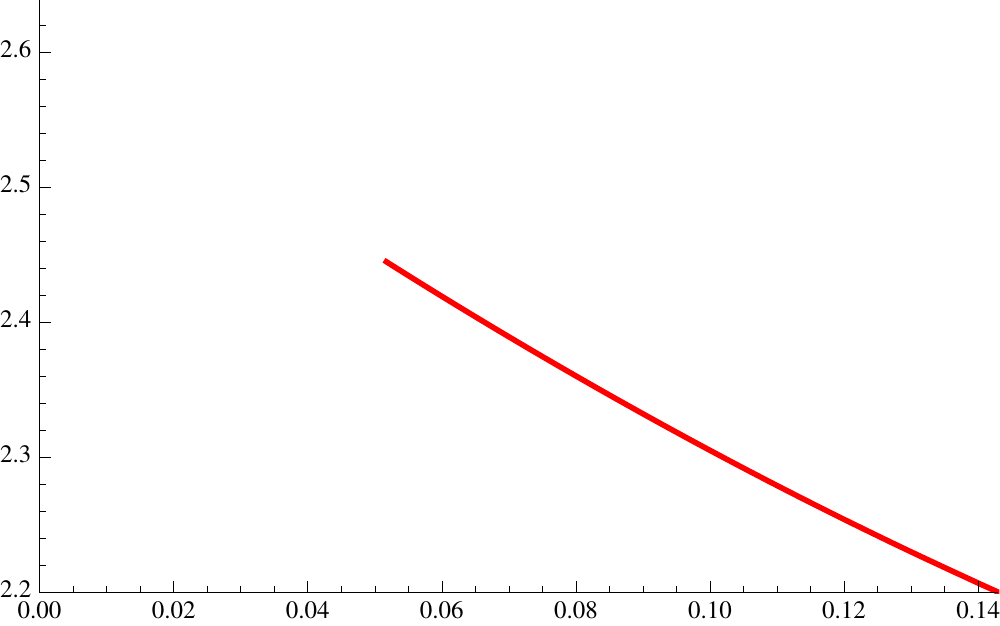}
\caption{Plot of ratios of coefficients against $\frac{1}{n^{2/3}}$.}
\label{fig:ratiosb}
\end{subfigure}

\begin{subfigure}[b]{0.49\textwidth}
\includegraphics[width=\textwidth]{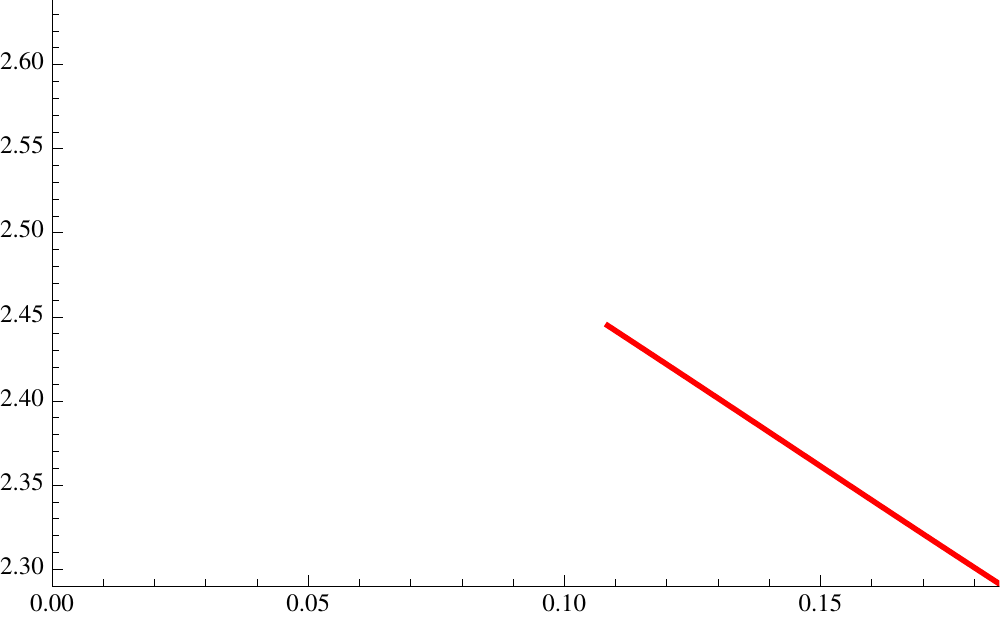}
\caption{Plot of ratios of coefficients against $\frac{1}{\sqrt{n}}$.}
\label{fig:ratiosc}
\end{subfigure}\hfill
\begin{subfigure}[b]{0.49\textwidth}
\includegraphics[width=\textwidth]{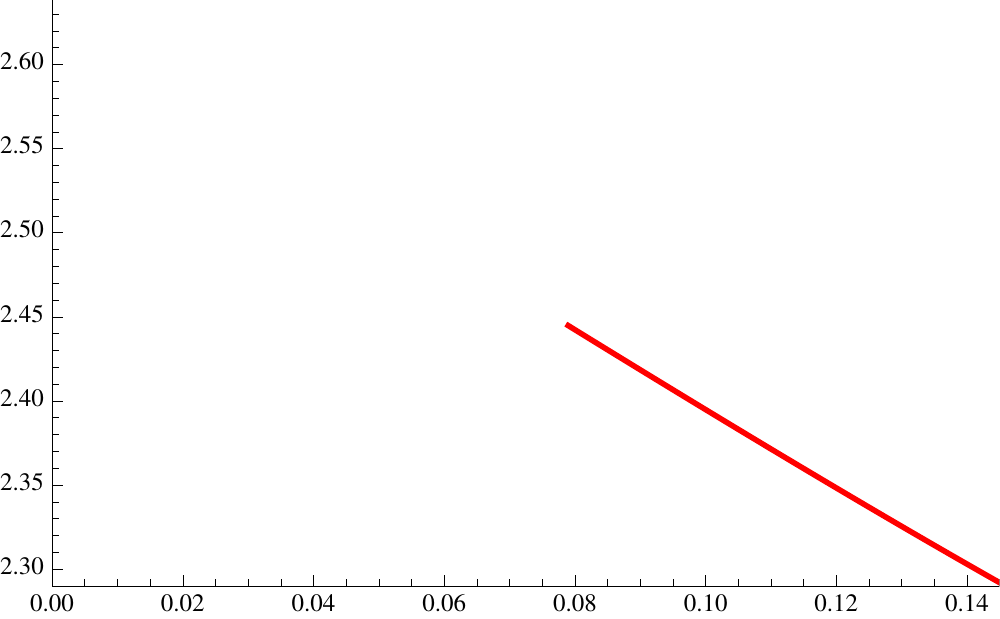}
\caption{Plot of ratios of coefficients against $\frac{1}{n^{4/7}}$.}
\label{fig:ratiosd}
\end{subfigure}
\caption{Bridge series ratios, $u=\log 2$. The top left-hand corner has ordinate $e^{\beta}$.}
\label{fig:ratios}
\end{figure}

The situation is similar for pushed SAWs (figures not shown), except that for SAWs the possible range of values for $\theta,$ as chosen by visual linearity of the ratio plots, is even wider, and we estimate $0.6 \le \theta \le 0.3,$ again consistent with the same expected value, $\theta=4/7 = 0.57142.. .$

In order to more directly  estimate the value of the exponent $\sigma$, we note from (\ref{eq:rn}) that $$(r_n\cdot e^{-\beta}-1) \sim const \cdot  n^{\sigma-1}.$$ The   plot of $\log(1-r_n \cdot e^{-\beta})$ against $\log {n}$ should be linear, with gradient $\sigma-1.$ To the naked eye it is, so we don't show it. However there is a small degree of curvature, so we extrapolate the local gradients, defined as
\BE \label{eq:locrat}
1-\sigma_n = \frac{\log \left ( 1 -{r_{n} \cdot e^{-\beta}} \right )-\log \left ( 1 -{r_{n-1}\cdot  e^{-\beta}} \right )}{\log {(n-1)} -\log(n)},
\EE
against $1/n.$ The results are shown in Figure~\ref{fig:sigmaa}.
The ordinates are estimators of $1-\sigma$, and it is clear that they are consistent with the conjectured $\sigma=3/7$.

A second estimator of $\sigma$ can be constructed from (\ref{eq:coef}) by noting that 
$$\log\left|\log \left(a_n e^{-\beta n}\right)\right| = \sigma \log{n} \left [1 + O\left ( \frac{\log{n}}{n^\sigma} \right ) \right ]. $$ As with the previous estimator, the   log-log plot of $\log \left(a_n e^{-\beta n}\right)$ against ${n}$ should be linear, with gradient $\sigma.$ We perform this calculation using bridges with $u=\log 2$. To the naked eye the plot is linear, so we don't show it. As in the previous case, there is a small degree of curvature, so we again extrapolate the local ratios, as defined above {\em mutatis mutandis}, and the resulting extrapolants are shown plotted against $1/n$ in Figure~\ref{fig:sigmab}. This is also entirely consistent with the expected value $3/7.$

\begin{figure}[t!]
\centering
\begin{subfigure}[b]{0.49\textwidth}
\includegraphics[width=\textwidth]{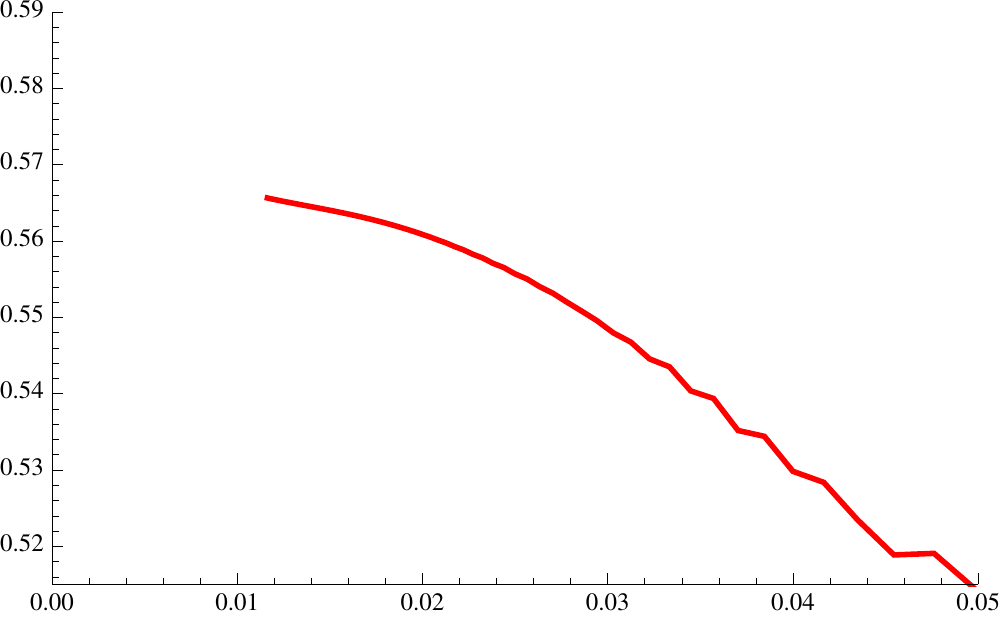}
\caption{Plot of estimates of $1-\sigma$ against $\frac{1}{n}$.}
\label{fig:sigmaa}
\end{subfigure}\hfill
\begin{subfigure}[b]{0.49\textwidth}
\includegraphics[width=\textwidth]{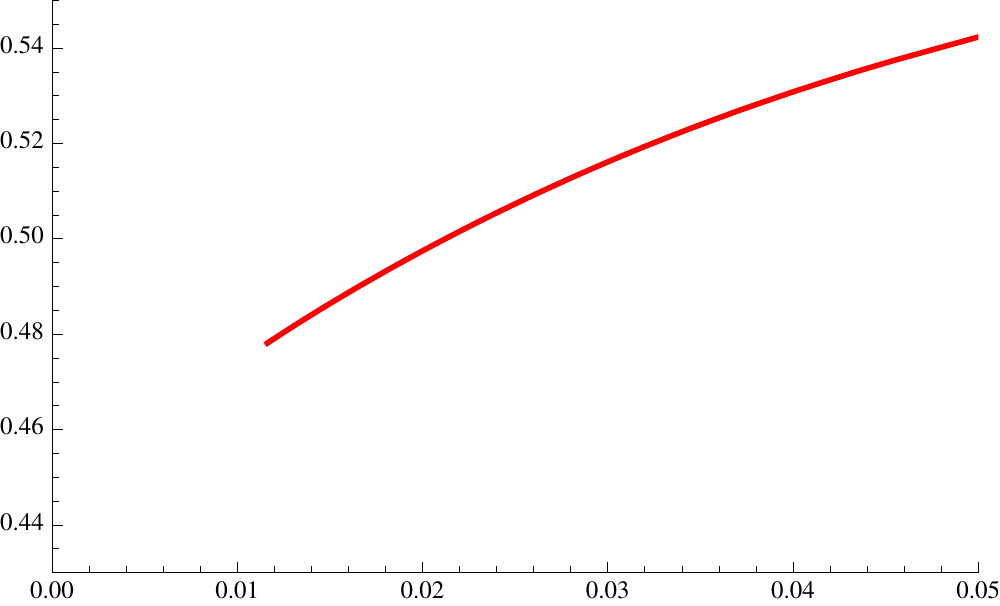}
\caption{Plot of estimates of $\sigma$  against $\frac{1}{{n}}$.}
\label{fig:sigmab}
\end{subfigure}
\caption{Two different estimators of exponent $\sigma.$ Both are consistent with $\sigma=3/7=0.42857\ldots.$}
\label{fig:sigma}
\end{figure}

We have repeated this calculation for bridges with other values of $u$, and find that this exponent estimate is numerically robust for $u$ in the range $e^{-u} \in [0.2,0.8].$ Outside this range we see crossover effects to other behaviour, as appropriate for $u=0$ and $u=\infty$. We have also repeated the calculation for SAWs with $u=\log 2$ (and other values of $u$), and find good numerical evidence for $\sigma \approx 0.43,$ supporting the expectation that $\sigma = 3/7$ exactly for both models.

\subsection{Estimation of $\mu_1.$}

We now take as known the value of the exponent $\sigma,$ and, to very high precision, the value of the leading growth constant $e^{\beta}.$ We next estimate the value of the ($u$-dependent) subdominant growth constant $\mu_1.$ Recall that, for random walks~\eqref{eqn:rw_asymps}, the $u$-dependence is given by 
\[\log{\mu_1} = -\lambda_u =  const \cdot u^{1-\sigma}.\]
(The same also occurs for compressed Dyck paths~\cite{G14a}.) We find compelling evidence for similar $u$-dependence in the case of both pushed bridges and pushed SAWs. We also find compelling evidence that, unsurprisingly, $\mu_1({\rm SAWs}) = \mu_1({\rm bridges}).$

We used several methods to estimate the value of $\mu_1.$ In all cases we utilised both the assumed value of the growth constant $e^{\beta}$ and the exponent $\sigma = 3/7.$ No one method was especially successful for all values of $u.$ As is common with series analysis, all methods involved some sort of extrapolation. It was frequently difficult to form a view as to the limiting extrapolant in some cases, usually due to a sequence of estimates having a turning point near the end of the range. Nevertheless, by using several methods, and looking at the spread between estimates, we were able to make moderately precise estimates.

The methods we used were the following:
(i) We first formed the logarithm of the ratio of successive coefficients (or, equivalently, the logarithm of the square root of the ratio of successive coefficients, $\sqrt(a_n/a_{n-2}),$ which reduces odd-even oscillations in the ratios). 
$$\log \left ( \frac{a_n}{a_{n-1}} \right ) - \log{e^{\beta}} = \frac{\sigma \log{ \mu_1}}{n^{1-\sigma}} + \frac{g}{n} + \frac{\sigma^2 \log^2{\mu_1}}{2n^{2-2\sigma}} + O \left ( n^{-2-\sigma} \right ),$$ in the case of both sequences $a_n = b_n$ and $a_n = c^+_n$. We calculated the series expansion of the left-hand side of the above equation, known to order $n_{\rm max}=86$ in the case of bridges, and $n_{\rm max}=60$ in the case of SAWs. We fitted the terms from order $n_{\rm min}$ to order $n_{\rm max}$ to the first two, and then the first three terms in the asymptotic expansion, using the Maple procedure \texttt{NonlinearFit}. In each case we steadily increased the value of $n_{\rm min}$ from $(n_{\rm max}-30)$ to $(n_{\rm max}-4).$ (Attempts using the first four terms were unsuccessful, due to lack of convergence).

(ii) The second method involved directly fitting to the coefficients. Since
$$\log{a_n} -n\cdot {\beta} = n^\sigma \log{\mu_1} + g\log{n} + const. +  O \left ( n^{-\sigma} \right ),$$ we could fit successive coefficient triples 
($a_{k-1},\,\, a_k \,\, a_{k+1}$), with $k$ increasing up to $n_{\rm max}-1,$ estimating the three unknowns, $\log{\mu_1}, \,\, g, \,\, const.$ from the three coefficients. We also fitted successive coefficient pairs ($a_{k-1},\,\, a_k $) to estimate the first two unknowns.

The third method (iii) used was a blend of the previous two methods, in that we fitted the logarithm of the ratio of successive coefficients to its asymptotic form by the second method -- that is, a fit using successive pairs or triple of coefficients.

The next two methods, (iv) and (v) involved estimating $\mu_1$ directly from the ratio of successive coefficients, as given by~\eqref{eq:rn}.
Since $$s_n=r_n \cdot e^{-\beta}-1 = \left (\frac{\sigma \log {\mu_1}}{n^{1-\sigma}} +\frac{g}{n}+\frac{\sigma^2 \log^2 \mu_1}{n^{2-2\sigma}}+ O\left (\frac{1}{n^{2-\sigma}}\right ) \right ),$$ a plot of $s_n$ against $n^{\sigma-1}$ should be linear, with slope $\sigma \log{\mu_1}.$
Slight curvature was evident, so we extrapolated the local gradients.

In an attempt to get more rapid convergence, we next, (v), eliminated the term $O(1/n)$ by forming the sequence
$$t_n = n\cdot s_n - (n-1)\cdot s_{n-1} = \left (\frac{\sigma^2 \log {\mu_1}}{n^{1-\sigma}} +\frac{\sigma^2(2\sigma-1) \log^2 \mu_1}{n^{2-2\sigma}}+ O\left (\frac{1}{n^{2-\sigma}}\right ) \right ).$$ As above, a plot of $t_n$ against $n^{\sigma-1}$ should be linear, with slope $\sigma^2 \log{\mu_1}.$ Again, slight curvature was evident, so we again extrapolated the local gradients.

The final method, (vi), we used was based on the observation that 
$$d_n=B_n(u)^{n^{-\sigma}} = 
\mu_1 \left (1 + \frac{g\log{n}}{n^\sigma} + O \left ( n^{-\sigma} \right ) \right ),$$
so $d_n$ provides a sequence of estimators of $\mu_1$ which can be extrapolated.

We show in Table \ref{tab:one} the estimates of $\log{\mu_1}$ obtained by the six methods described. The actual analysis comprises dozens of pages, so we only give a summary. Rather than attempt to give precise error bars, which in any series analysis are perhaps better described as confidence limits, we consider it more meaningful to quote parameters to a level of precision such that we expect uncertainties to be confined to the last quoted digit. From the various estimates, we first ignore obvious outliers, then simply take the average of the remaining estimates, and give this as our combined estimate. Comparing results for SAWs and bridges, it appears that $\mu_1$ is the same, up to uncertainty in the quoted best estimates.

A little numerical experimentation shows that the $u$-dependence is indistinguishable from $\log\mu_1 \propto u^{4/7}.$ We note that, like random walks~\eqref{eqn:rw_asymps} and Dyck paths~\cite{G14a}, this matches $\log\mu \propto u^{1-\sigma}$. Fitting the constant of proportionality to the combined estimates leads us to the result
$$\log{\mu_1} \approx -2.62u^{4/7}.$$ Experimentation with data sets for other $u$-values is consistent with this result. In order to compare with the estimates given below, this formula gives $\log{\mu_1} = -2.913, \,\, -2.125 \,\, -1.453$ for $e^{-u}=0.3, \,\, 0.5, \,\, 0.7$ respectively.

\begin{table}[htbp]
  \begin{center}
\begin{tabular}{|p{100pt} |p{60pt} p{60pt} p{60pt}|}
\hline
Method & $e^{-u}=0.3$ & $e^{-u}=0.5$ & $e^{-u}=0.7$ \\ \hline
(i) Bridges, 2 terms &$ > -2.98 $ & $ \approx -2.177 $ & $ > -1.51$ \\
(i) Bridges, 3 terms &$ \approx -2.89 $ & $ > -2.180 $ & $ < -1.30$ \\
(ii) Bridges, 2 terms &$-3.0  $ & $ -2.15 $ & $-1.5  $ \\
(ii) Bridges, 3 terms &$ -2.92 $ & $ n.c. $ & $ -1.46$ \\
(iii) Bridges, 2 terms &$-2.905 $ & $ -2.13 $ & $ -1.47$ \\
(iii) Bridges, 3 terms &$ n.c. $ & $ n.c. $ & $-1.46 $ \\
(iv) Bridges & $ -2.96 $&$-2.14$& $-1.47  $ \\
(v) Bridges & $ -2.92 $&$-2.14$& $n.c.  $ \\
(vi) Bridges & $ -3.0 $&$-2.21$& $ -1.5 $ \\
Combined Estimate & -2.92 & -2.14 & -1.465 \\
\hline
(i) SAWs, 2 terms &$ > -3.00 $ & $ >-2.21$ & $ > -1.55$ \\
(i) SAWs, 3 terms &n.c. & $ \approx -2.14 $ & $ > -1.485$ \\
(ii) SAWs, 2 terms &$ >-2.98 $ & $-2.13  $ & $ -1.5$ \\
(ii) SAWs, 3 terms &$\le -2.90.  $ & $ n.c. $ & $-1.45 $ \\
(iii) SAWs, 2 terms &$-2.87  $ &$-2.135  $ & $-1.47  $  \\
(iii) SAWs, 3 terms &$-2.92.  $ & $n.c.  $ & $ -1.45$ \\
(iv) SAWs & $ -2.84 $&$-2.12$& $ -1.42 $ \\
(v) SAWs & $ -2.89 $&$-2.15$& $ -1.47 $ \\
(vi) SAWs & $n.c.  $&$n.c$& $ n.c. $ \\
Combined Estimate & -2.91 & -2.13 & -1.45 \\
\hline
 \end{tabular}
  \caption{Estimates of $\log{\mu_1}$ for SAWs and bridges for various $e^{-u}$ values, 
  using the six methods described. ($n.c.$ means no convergence).}
  \label{tab:one}
\end{center}
\end{table}

\subsection{Estimation of the exponent $g.$}

In the previous section we estimated the value of $\mu_1$ from the series assuming the values of $e^{\beta}$ and $\sigma.$ Attempting to do the same in order to estimate the exponent characterising the sub-sub-dominant term $n^g$ making the same assumptions was unsuccessful.  That is to say, we could not come up with any method of analysis that gave a consistent estimate of $g$ for the SAW or bridge case. Note that we expect $g$ to be $u$-independent, but not necessarily to be the same for SAWs and bridges. 

However if we accept the conjecture (above) that $\mu_1({\rm SAWs})=\mu_1({\rm bridges}),$ one can estimate $g_{\rm SAWs}-g_{\rm bridges}$ by studying the asymptotics of the series formed by taking the term by term quotient of the SAW and the bridge series. Defining $v_n=c^+_n/b_n,$
we expect from the asymptotics that $v_n \sim const \cdot  n^\alpha,$ where $\alpha = g_{\rm SAWs}-g_{\rm bridges}.$

For $e^{-u}$ in the range $[0.1,0.6]$ we find estimates of $\alpha,$ obtained by calculating the gradient of a simple ratio plot of $r_v(n)=v_n/v_{n-1}$ against $1/n$, to lie consistently in the range $0.61$ to $0.63.$ For $e^{-u} \ge 0.7$ the exponent estimators  clearly have a turning point  at a value beyond the number of coefficients we have available. Accordingly, we cannot estimate $\alpha$ in this $u$-regime.

In order to estimate the individual exponents, we re-analyse the original series using not only our best estimate of $e^{\beta}$ and the value $\sigma=3/7,$ but also the estimate of $\mu_1$ obtained in the previous section. We can, using these values, divide the original series coefficients by $e^{n \beta} \cdot \mu_1(u)^{n^\sigma}.$ The resultant coefficients should behave asymptotically as $const  \cdot n^g.$ Thus $g$ can be estimated in a variety of ways, but given the uncertainty in the value of $\mu_1,$ there is no point in using sophisticated methods. Rather, we just look at the gradient of the ratio plot of the divided coefficents, when plotted against $1/n.$ The ratio of succesive coefficients in this case should behave as $1+g/n +o(g/n)$. Thus plotting the ratios against $1/n$ should give a straight line with gradient $g$.

In this way we estimate $g_{\rm SAWs} \approx 0.15, \,\, 0.15, \,\, 0.16$ for $e^{-u}=0.3, \,\, 0.5, \,\, 0.7$ respectively, and $g_{\rm bridges} \approx -0.5, \,\, -0.5, \,\, -0.44$ for $e^{-u}=0.3, \,\, 0.5, \,\, 0.7$ respectively. These results are reasonably consistent with the preceding estimate that the gap between the exponents should be in the range $0.61$ to $0.63.$ Our estimate for $g_{\rm bridges}$ has an uncertainty of about $\pm 0.05,$ while that of $g_{\rm SAWs}$ is $\pm0.02$, so our estimates have a difference of $0.65 \pm 0.07,$ which is reasonably consistent with the direct estimate found above.

The heuristic arguments in the previous section predict that $g_{\rm bridges} = -\frac{13}{28} = -0.4642857\ldots,$ in agreement with our analysis above, while $g_{\rm SAWs}=\frac{3}{16} = 0.18750,$ again in agreement with the above analysis. Similarly, $\alpha = \frac{73}{112}=0.651785\ldots,$ as against our direct estimates of $\alpha \approx 0.65$ and $\alpha \approx 0.62.$

\section{Conclusion}\label{sec:conclusion}

We have studied several models of self-avoiding walks (SAWs) and polygons (SAPs) which are constrained to lie in the upper half-plane and are subjected to a compressive force. The force is applied to the vertex or vertices of the walk located at the maximum distance above the boundary of the half-space. We have in particular focused on three types of objects: SAWs, SAPs and self-avoiding bridges. In each case, we have considered the partition function of objects of size $n$, with the aim of determining the asymptotic behaviour of these partition functions in the $n\to\infty$ limit.

We used the conjectured relation with the Schramm-Loewner evolution to predict the asymptotic forms of the partition functions, including the values of the exponents. These values, as stated in Prediction~\ref{pred:main}, are that (for SAWs, bridges and SAPs respectively)
\begin{align*}
\Cmax_n(u) &\sim A^+_u \, n^{3/16} \, \exp\{-\lambda_1 u^{4/7} \, n^{3/7}\},\\
B_n(u) & \sim A^b_u \, n^{-13/28} \, \exp\{-\lambda_1 u^{4/7} \, n^{3/7}\},\\
\Pmax_n(u) &\sim A^{p+}_u\, n^{-11/7}\, \exp\{-\lambda_2 u^{4/7}\,n^{3/7}\},
\end{align*}
where $u>0$ corresponds to the compressive force regime, $\lambda_1,\lambda_2>0$ are constants and $A^+_u, A^b_u, A^{p+}_u$ are functions of $u$.

Finally, we tested the predictions for SAWs and bridges by analysing exact enumeration data and found them to be in agreement (within the range of uncertainty resulting from analysis of limited series data). A similar analysis for SAPs has been performed elsewhere.

\section*{Acknowledgements}

This work was supported by an award (IJ) under the Merit Allocation Scheme on the NCI National Facility at the ANU
and  by  funding under the Australian Research Council's Discovery Projects scheme by the grant
DP140101110 (AJG and IJ).
Gregory Lawler is supported by National Science Foundation
grant DMS-0907143. Nicholas Beaton is supported by the Pacific Institute for the Mathematical Sciences.

\bibliographystyle{amsplain}


\appendix

\section{Asymptotics of an integral}

Several times in this paper, we used the asymptotics 
of a particular integral.  Here we justify that relation.
\begin{lem}\label{lem:integralest}
If $\alpha,b, k >0,r \in \R$, then as $m \rightarrow \infty$, 
\begin{equation}\label{integralest}
 \int_0^\infty y^r \, \exp \left\{- \left(\frac{km }{\alpha (by)^\alpha} + by\right) \right\} \,dy \sim   C \, m^{\frac {2r+1}{2(1+\alpha)}}
 \, \exp \left\{- \lambda \,m ^{\frac 1{1+\alpha}} \right\}
   , 
   \end{equation}
where
\[  C =  
\,  b^{-(1+r)} \, k^{ \frac {2r+1}{2(1+\alpha)}} \sqrt {\frac{2\pi  }
 { 1+\alpha  }},  \qquad
  \lambda  = \frac{ (\alpha + 1) \, k^{\frac{1}{1+\alpha}}}
    {\alpha}
     =  \frac{ (\alpha + 1)}{\alpha} \, \left(\frac{k}{b^\alpha}
     \right)^{\frac{1}{1+\alpha}}   \, b^{\frac{\alpha}{1 + \alpha}}.  \]
\end{lem}

\begin{proof}
Using the change of variables $x = by$, we see it suffices
to prove the result for $b=1$.  We write $km = n^{1 + \alpha}$,
so it suffices to show that 
\[  \int_0^\infty x^r \,  \exp \left\{- \left(\frac{n^{ 1 + \alpha }}{\alpha x^\alpha} + x \right) \right\} \,dx = 
      \sqrt {\frac{2\pi }{1+\alpha}} \, n^{r + \frac 12}
    \, e^{-n(\frac{\alpha + 1}{\alpha})} \, \left[1+O (n^{-1/2})
     \right].\]
Using the substitution $x = n^{1/2} y + n$,
 \begin{eqnarray*}
 \lefteqn{  \int_0^\infty x^r \, \exp \left\{- \left(\frac{n^{ 1 + \alpha }}{\alpha x^\alpha}
 + x \right) \right\} \,dx} \hspace{1in} \\
  & = & n^{r + \frac 12}  \int_{-\sqrt n}^\infty
    \left(1 + \frac{y}{\sqrt n}
    \right)^{r} \exp\left\{-n\, \left(\frac {n^{  \alpha}}
    {\alpha (n^{1/2} y + n)^\alpha} + \frac y{\sqrt n} + 1
      \right) \right\} \, dy \\
  & = & n^{r + \frac 12} \, e^{-n(1+ \frac 1\alpha)}
   \int_{-\sqrt n}^\infty \left(1 +(y/\sqrt n)\right)^r\,
    \exp\left\{-n\,g(y/\sqrt n)
    \ \right\} \, dy ,
 \end{eqnarray*}
  where
  \[  g(t) = g_\alpha(t) =  \alpha^{-1} \left((t+1)^{-\alpha} - 1\right) + t.\]
  Using the substitution $t = y/\sqrt n$,
 \[ n^{-(r + \frac 12 ) } \, e^{n(1 + \frac 1 \alpha)}
   \int_0^\infty x^r\,\exp \left\{- \left(\frac{n^{ 1 + \alpha }}{\alpha x^\alpha}
 + x \right) \right\} \,dx =  \sqrt n \, \int_{-1}^\infty 
 (1+t)^r \, \exp\left\{-n\, g(t)
     \right\} \, dt .\]
Since $g(0) = 0, g'(0) = 0, g''(0) = \alpha +1$,
we  have  
 \[   g(t) = \frac {\alpha + 1} 2 \, t^2 + O(t^3), 
 \]
 which is valid for $|t| \leq 1/2$. Note that there exists
 $c> 0$ such that
\[   g(t) \geq c\, n^{-2/3} , \;\;\;\ |t| \geq n^{-1/3} , \]
\[    g(t) \geq   t - \alpha^{-1}  \;\;\;\; t \geq 0 ,\]
\[    g(t) \geq \frac{1}{(1+t)^\alpha}
  - \frac{\alpha + 1}{\alpha} , \;\;\;\; -1 < t \leq - \frac 12.\]
Using these estimates in order, we see that 
\[    \int_{-1 +n^{-1}}^{2/\alpha}(1+t)^{r}
 \exp\{-n g(t)\} \, \mathbbm 1_{\{|t| \geq n^{-1/3}\}}
  \, dt =  O\left( n^{|r|} \,
  \exp\left\{- c n^{1/3}\right\} \right) = o(n^{-1/2}), \]
  \[   \int_{2/\alpha}^\infty (1+t)^{r}\exp\{- n g(t)\} \, dt
  \leq   e^{n/\alpha}
   \, \int_{2/\alpha}^\infty (1+t)^{r}\exp\{- n t\} \, dt
    \leq  O(n^{|r|} \, e^{-n/\alpha})  = o(n^{-1/2}), \]
  \[   \int_{-1} ^{-1 + n^{-1}}(1+t)^{r}
 \exp\{-n g(t)\} 
  \, dt =  o(n^{-1/2}).\]
Hence, 
\[  
\sqrt n \, \int_{-1}^\infty  (1+t)^{r}\,
 \exp\left\{-n\, g(t)
     \right\} \, dt  =   o(n^{-1/2}) + 
\sqrt n \, \int_{-n^{-1/3}}^{n^{-1/3}} (1+t)^{r}\,
 \exp\left\{-n\, g(t)
     \right\} \, dt.\]
For $|t| \leq n^{-1/3}$, we write
\begin{align*}
 (1 + t)^r \,  e^{-n g(t)} & = [1+ O(t)]\, \exp \left\{- n
 \left(\frac{\alpha + 1}{
2}\, t^2  + O(t^3)\right)\right\}\\
& = 
   e^{-n(\alpha +1)t^2/2} \, [1 + O(t) +  O(nt^3)].
  \end{align*}
Here we have restricted to $|t| \leq n^{-1/3}$ so that
$nt^3 = O(1)$ and hence the expansion
$ e^{O(nt^3)} = 1 + O(nt^3)  $ is valid.
 
Then
\[ 
\sqrt n \, \int_{-n^{-1/3}}^{n^{-1/3}}  
  e^{-n(\alpha +1)t^2/2}  \, dt    = 
  \int_{-n^{1/6}}^{n^{1/6} }
     e^{-(\alpha +1)s^2/2} \, ds 
      =   \sqrt{\frac{2\pi}{(\alpha +1)}} 
    -o(n^{-1/2}),
\]
 and 
\[  
 \sqrt n \, \int_{-n^{-1/3}}^{n^{-1/3}}
  \left(|t| + n|t|^3\right) \, e^{-n(\alpha +1)t^2/2}  \, dt 
  \leq     \int_{-\infty}
   ^{\infty}  \frac{|s| + |s|^3}{\sqrt n}  
    \,  e^{-(\alpha +1)s^2/2}  \, ds  = O(n^{-1/2}) . \]
    \end{proof}
    
\end{document}